\documentclass[12pt, a4paper, fleqn]{article}

\usepackage{etex}
\usepackage[utf8]{inputenc}
\usepackage[TS1,T1]{fontenc}
\usepackage{xspace}
\usepackage{array}
\newcolumntype{P}[1]{>{\centering\arraybackslash}p{#1}}
\usepackage{multirow}
\usepackage{tikz-cd}
\usetikzlibrary{external}
\usepackage{smartdiagram}
\usepackage{tensor,tabularx}
\usetikzlibrary{decorations.markings}
\usepackage{layout} 
\usetikzlibrary{arrows.meta} 
\usepackage{pdfpages}
\usepackage{float}

\usepackage[a4paper,tmargin=2.75truecm,bmargin=2.5truecm,hmargin=2.8truecm,  rmargin=2.9truecm, lmargin=2.9truecm, marginparwidth=25mm]{geometry}
\usepackage{tocloft}
\usetikzlibrary{arrows,matrix,positioning,shapes,arrows}
\usetikzlibrary{shapes.geometric, arrows, calc, intersections}
\newcommand{\tikznode}[2]{\relax
	\ifmmode%
	\tikz[remember picture,baseline=(#1.base),inner sep=0pt] \node (#1) {$#2$};
	\else
	\tikz[remember picture,baseline=(#1.base),inner sep=0pt] \node (#1) {#2};%
	\fi}

\usepackage[english]{babel}
\usepackage{blindtext} %This package generates automatic text
\usepackage{epigraph} 
\usepackage[]{libertinus-type1}
\usepackage{amsfonts, mathtools,slashed,mathdots,bm}
\usepackage{libertinust1math}
\usepackage[bb=ams]{mathalpha}
\makeatletter
\newcommand{\addbar@}[3]{%
	\makebox[0pt][l]{%
		\raisebox{#1}[0pt][0pt]{%
			\kern#2
			\scalebox{#3}[0.8]{$\m@th\mathchar"84$}%
		}%
	}%
}
\DeclareRobustCommand{\lambdabar}{\text{\addbar@{0.1ex}{0.18em}{1}}\lambda}
\makeatother
\usepackage{slashed}
\usepackage{enumitem}
\newlist{enum-hypothesis}{enumerate}{1}
\setlist[enum-hypothesis]{label=(\arabic*),itemsep=0pt, parsep=0pt}

\setlist[enumerate,1]{label=\arabic*., ref=\arabic*, topsep=1pt, itemsep=2pt, parsep=0pt, leftmargin=1.5em, itemindent=0em, labelsep=0.2em, labelwidth=1.3em}
\setlist[enumerate,2]{label=\alph*., ref=\theenumi.\alph*, topsep=1pt, itemsep=2pt, parsep=0pt, leftmargin=0.5em, itemindent=0em, labelsep=0.2em, labelwidth=1.5em}
\setlist[enumerate,3]{label=\roman*., ref=\theenumii.\roman*, topsep=1pt, itemsep=2pt, parsep=0pt, leftmargin=0.5em, itemindent=0em, labelsep=0.2em, labelwidth=1.2em}

\usepackage{booktabs, dcolumn}
\usepackage{graphicx, subfig}

\usepackage{pbox}

\usepackage[square,numbers,compress,merge]{natbib}
\usepackage[thmmarks,amsmath]{ntheorem}
\newtheorem{theorem}{Theorem}[section]
\newtheorem{proposition}[theorem]{Proposition}

\newtheorem{definition}[theorem]{Definition}

\theoremsymbol{$\square$}
\theoremstyle{plain}
\theorembodyfont{\upshape}
\newtheorem{remark}[theorem]{Remark}
\theoremsymbol{$\Diamond$}

\theoremsymbol{}
\theoremstyle{break}
\theorembodyfont{\slshape}

\theoremheaderfont{\scshape}
\theorembodyfont{\upshape} 
\theoremstyle{nonumberplain}
\theoremseparator{} 
\theoremsymbol{\rule{1.3ex}{1.3ex}} 
\newtheorem{proof}{Proof}

\usepackage{hyperref}

\hypersetup{
	plainpages=false,
	colorlinks=true,% à mettre avant la suite : supprime l'encadrement des liens
	linkcolor=blue, 
	anchorcolor=black, 
	citecolor=blue, 
	urlcolor=blue, 
	menucolor=black, 
	filecolor=black, 
	bookmarksopen=true,
	bookmarksnumbered=true}

\hypersetup{
	pdfauthor={Gaston Nieuviarts},
	pdftitle={Titre},
	pdfsubject={},
	pdfcreator={pdflatex},
	pdfproducer={pdflatex},
	pdfkeywords={}
}

\newcommand{\Man}{\mathcal{M}}
\newcommand{\CM}{C^\infty(\Man)}

\newcommand{\SpinBun}{S} % spinor bundle

\newcommand{\bbbone}{{\text{\usefont{U}{bbold}{m}{n}\char49}}} % from \DeclareSymbolFont{bbold}{U}{bbold}{m}{n} in mathbbol.sty
\newcommand{\fbb}{\bbbone_F}
\newcommand{\mbb}{\bbbone_{2^{m}}}

\newcommand{\calA}{\mathcal{A}}
\newcommand{\calB}{\mathcal{B}}
\newcommand{\calC}{\mathcal{C}}

\newcommand{\calG}{\mathcal{G}}
\newcommand{\calH}{\mathcal{H}}
\newcommand{\calI}{\mathcal{I}}

\newcommand{\calK}{\mathcal{K}}
\newcommand{\calL}{\mathcal{L}}
\newcommand{\calM}{\mathcal{M}}

\newcommand{\calO}{\mathcal{O}}

\newcommand{\calS}{\mathcal{S}}

\newcommand{\calU}{\mathcal{U}}

\newcommand{\algA}{\calA}

\newcommand{\defeq}{\vcentcolon=} % :=
\DeclareMathOperator{\ad}{ad}
\DeclareMathOperator{\Aut}{Aut}
\DeclareMathOperator{\card}{card}

\DeclareMathOperator{\Diff}{Diff} %% Imagesq

\DeclareMathOperator{\Inn}{Inn}

\DeclareMathOperator{\tr}{tr}	   %% trace
\DeclareMathOperator{\Tr}{Tr}	   %% trace

\newcommand{\DM}{D_\Man}

\newcommand{\lag}{{\calL}} %Lagrangien
\newcommand{\act}{\calS} %Lagrangien
\newcommand{\US}{S} %Unit sphere
\newcommand{\Sp}{\mathcal{S}} % SPin structure
\newcommand{\Dir}{D} 

\newcommand{\inn}{{\langle\, \cdot  \, , \, \cdot  \, \rangle}}

\newcommand{\tw}{K}
\newcommand{\inntw}{{\langle\, \cdot \, , \, \cdot  \, \rangle_\tw}}

\newcommand{\Tadj}{{\dagger_{\tiny \tw}}}

\newcommand{\adp}{{\dagger_{\tiny p}}}

\newcommand{\KDir}{\Dir^\tw}

\newcommand{\FDir}{\Dir_F}
\newcommand{\JF}{J_F}
\newcommand{\JH}{{\hat{J}}}
\newcommand{\CH}{{\hat{C}}}

\newcommand{\GF}{\Gamma_F}
\newcommand{\kgm}{\gamma_\tw}
\newcommand{\tgm}{\tilde\gamma}

\newcommand{\tm}{{T\Man}}
\newcommand{\txm}{{T_x\Man}}

\newcommand{\twx}{\tw_x}
\newcommand{\AF}{\calA_F}
\newcommand{\HF}{\calH_F}

\newcommand{\utw}{U_\tw}

\newcommand{\g}{g} 
\newcommand{\tg}{{\tilde{\g}}}

\newcommand{\gr}{\g_{\scriptscriptstyle R}} 
\newcommand{\gk}{\g_{\scriptscriptstyle \tw}} 

\newcommand{\tad}{{\widetilde{\ad}}} 
\newcommand{\tphi}{{\widetilde{\phi}}} 

\newcommand{\ovg}{\calO(V, \g)} 
\newcommand{\opvg}{\calO_{\bullet}(V, \g)} 
 
\newcommand{\gl}{{GL}} 
\newcommand{\glv}{\gl(V)} 
\newcommand{\glpv}{\gl_{\bullet}(V)}

\newcommand{\cl}{c}
\newcommand{\tcl}{\tilde{c}}

\newcommand{\twphi}{{ \Phi^{\,\tw}}}

\newcounter{mnotecount}[section]
\renewcommand{\themnotecount}{\thesection.\arabic{mnotecount}}
\newcommand{\mnote}[1]%
{\protect{\stepcounter{mnotecount}}${}^{\text{\footnotesize$\bullet$\themnotecount}}$%
	\reversemarginpar%
	\marginpar{\raggedleft\footnotesize$\bullet$\themnotecount: #1}}

\newlength{\mnotewidth}
\definecolor{blueamu}{RGB}{0, 101, 189}
\definecolor{cyanamu}{RGB}{61, 183, 228}
\newcommand{\dhorline}[3][0]{%
	\tikz[baseline=-2pt]{\path[decoration={markings, 
			mark=between positions 0 and 1 step 2*#3
			with {\node[color=blueamu, fill, circle, minimum width=#3, inner sep=0pt, anchor=south west] {};}},postaction={decorate}]  (0,#1) -- ++(#2,0);}}
\newcommand{\dvertline}[3][0]{%
	\tikz[baseline=2em]{\path[decoration={markings,
			mark=between positions 0 and 1 step 2*#2
			with {\node[color=black!50, fill, circle, minimum width=#2, inner sep=0pt, anchor=south west] {};}},postaction={decorate}] (0, #1) -- ++(0,#3);}}
\makeatletter\newcommand\HUGE{\@setfontsize\Huge{28}{0}}\makeatother		
\numberwithin{equation}{section}

\allowdisplaybreaks
\begin{document}
	\renewcommand\figurename{Fig.}
	
	%%%%%%%%%%%%%%%%%%%%%%%%%%%%
	%% VERSION NON JOURNAL
	{% beginning of local redefinition
		\makeatletter\def\@fnsymbol{\@arabic}\makeatother % redefinition of \@fnsymbol
 
		\title{Emergence of Lorentz symmetry\\ from an almost-commutative\\ twisted spectral triple.}
		%	A duality between\\ Lorentzian twisted spectral triples\\ and spectral triples
		
		%A duality between spectral triples and\\ pseudo-Riemannian spectral triples\\ through the twist
		
		\author{G. Nieuviarts\footnote{gaston.nieuviarts.wk@gmail.com}\\
			{\small  }%
			%\footnote{thierry.masson@cpt.univ-mrs.fr, gaston.nieuviarts@etu.univ-amu.fr}%
			\\
			\small{ }\\[2ex]
			%			{\small DIMA, Universita di Genova}%
			%		%\footnote{thierry.masson@cpt.univ-mrs.fr, gaston.nieuviarts@etu.univ-amu.fr}%
			%		\\
			%		\small{via Dodecaneso, 16146 Genova, Italia}\\[2ex]
		}
	 
		\date{}
		\maketitle
		\vspace{-\baselineskip}
		\vspace{-\baselineskip}
		\vspace{-\baselineskip}
	}  
	%% END OF VERSION NON JOURNAL
	%%%%%%%%%%%%%%%%%%%%%%%%%%%%
	
	\setcounter{tocdepth}{2}
	% profondeur table 
\vspace{-0.99cm}
	\begin{abstract}
		
		This article demonstrates how the transition from a (Riemannian) twisted spectral triple to a pseudo-Riemannian spectral triple arises within an almost-commutative spectral triple. This opens a new perspective on the Lorentzian signature problem, showing that the almost-commutative structure at the heart of the noncommutative standard model of particle physics could be the origin of the emergence of a Lorentzian spectral triple, starting from self-adjoint Dirac operators and conventional inner product structures, in the framework of twisted spectral triples. We present an alternative to Wick rotation, acting on the metric and the Christoffel symbols in a way that does not introduce any complex numbers. 
		
	\end{abstract}

	\tableofcontents 
	\vspace{-0.2cm}
	\newpage

	\section{Introduction}
	\label{sec introduction}
	
	Noncommutative geometry is a branch of mathematics that generalizes classical geometry to spaces where the algebra of functions is no longer commutative. Inspired by quantum mechanics, where observables (elements in an algebra) do not commute, this theory provides tools to describe "spaces" that are too irregular or singular to be treated with traditional differential geometry. One key realization of noncommutative geometry is its ability to recover Riemannian manifolds through spectral data. A. Connes, a pioneer of the field, developed the concept of spectral triples to encode the geometric information. A spectral triple $(\algA, \calH,\Dir)$ is the data of an involutive unital algebra $\algA$ represented by bounded operators on a Hilbert space $\calH$ together with a self-adjoint operator $\Dir$ acting on $\calH$. In this framework, Riemannian geometry can be reconstructed from the spectral properties of the Dirac operator by starting from the spectral triple 
	$(C^\infty(\Man), L^2(\Man, \Sp), \Dir, J, \Gamma)$, where $L^2(\Man, \Sp)$ is the Hilbert space of square integrable sections of the spinor bundle over the smooth manifold $\Man$. The Dirac operator $\Dir$ is constructed using the Riemannian gamma matrices and the Levi-Civita spin connection. $J$ is a given anti-unitary operator, and $\Gamma$ is the grading operator.

	In theoretical physics, fundamental interactions are described as Yang-Mills and Gravity theories, both based on the concept of connection on vector bundles. The extension of these concepts to noncommutative geometry relies on a crucial duality given by the Serre-Swan theorem. The gauge group emerges from inner automorphisms of $\algA$. According to the noncommutative differential structure we choose, two main frameworks have been used to construct Noncommutative Gauge Field Theories. The first one (historically), proposed by M. Dubois-Violette, R. Kerner, and J. Madore in \cite{DuboKernMado90b} concerns the use of derivation-based differential structure. The second one, mainly developed by A. Connes, J. Lott, and A. Chamseddine in \cite{ConnLott90a,chamseddine1997spectral} concerns the use of spectral triples based differential structure. Lets focus on this latter approach by considering a spectral triple $(\calA,\, \calH,\, \Dir,\, J,\, \Gamma)$ (see \cite{connes2006quantum, connes2013spectral} for the complete definition). The derivative is given by $\delta(.)=\left[D,.\right]$ so that one-forms are given by elements $A=\sum_{k} a_{k}\left[D, b_{k}\right]$ with $a_{k}, b_{k} \in \algA$. The connection with gauge theory goes as follows, if we consider a unitary element $u \in \calU(\algA)$ (the set of unitaries in $\algA$), and define the adjoint unitary $U = \pi(u) J \pi(u) J^{-1}$, then, considering a one-form $\omega$ as a gauge field and the associated “fluctuated” Dirac operator $\Dir_\omega := \Dir + \omega + \epsilon' J \omega J^{-1}$, inner fluctuations  $(\Dir_\omega)^u:=U\Dir_\omega U^\dagger$ are then equivalent to gauge transformations $\omega^u := u \omega u^\dagger+ u \delta(u^\dagger)$ via the relation $(\Dir_\omega)^u = \Dir_{\omega^u}$. The spectral action, a central spectral invariant, is given by $\act_b[\Dir]:= \Tr f(\Dir \Dir^\dagger / \Lambda^2)$ (for more details, see \cite{chamseddine2007gravity}).

	Based on this construction, we can consider a noncommutative extension of $(C^\infty(\Man), L^2(\Man, \Sp), \Dir, J, \Gamma)$ by taking the Almost-Commutative (AC) algebra $\calC^\infty(\Man) \otimes \AF$ with $\AF$ a finite dimensional noncommutative algebra. The Dirac operator is $\DM \otimes \bbbone + \Gamma \otimes \FDir$, with Hilbert space $L^2(\calM,\SpinBun) \otimes \HF$, real structure $J\otimes  \JF$ and grading $\Gamma \otimes \GF$. A fundamental aspect of this AC manifold is that its spectral action corresponds to the Lagrangian of an Einstein-Yang-Mills-Higgs model, entirely determined by the product spectral triple. Driven by the constraints coming with the axioms of even real spectral triples, a reformulation of the Standard Model of Particle Physics coupled to gravitation was elaborated, with crucial steps in \cite{connes1989particle, chamseddine1997spectral} until its full development in \cite{chamseddine2007gravity} (we refer to \cite{connes2006quantum} for detailed explanations). At the heart of this construction lies the finite geometry encoded in $\calA_F\defeq \mathbb{C}\oplus \mathbb{H} \oplus M_3(\mathbb{C})$ and its associated spectral triple. The symmetry group of the spectral action is given by $\calG=Map(\Man, G) \rtimes \Diff(\Man)$, where $\Diff(\Man)$ denotes the diffeomorphism group, and $Map(\Man, G)$ represents the gauge group of the second kind, with $G = U(1)\times SU(2)\times  SU(3)$. In this way, noncommutative geometry provides a unified framework to describe both Einstein-Hilbert gravity (in Euclidean signature) and classical gauge theories. This framework is known as the NonCommutative Standard Model. Notably, it offers an elegant description of the Standard Model, including the Higgs mechanism and neutrino mixing, as "gravity" on an AC-manifold. The fermionic masses are encoded into $\FDir$, so that the masses of the Higgs boson and the ones of fermions became related, thus providing a prediction for the Higgs mass.

	The noncommutative approach links geometry to operator algebras, revealing new mathematical and physical insights. But it falls short when describing spacetimes with pseudo-Riemannian signatures that are essential in physics. Indeed, realistic models require a formulation that incorporates time and causality, motivating the search for Lorentzian approaches within noncommutative geometry. This issue has spurred various proposals, from modified axioms for spectral triples to alternative formulations that attempt to reconcile causality with noncommutative structures see \cite{barrett2007lorentzian, franco2014temporal, paschke2006equivariant, strohmaier2006noncommutative, bochniak2018finite, devastato2018lorentz, barrett2007lorentzian, van2016krein, bizi2018space, d2016wick}. An essential feature of most of these approaches is the replacement of the Hilbert space by a Krein space, turning the inner product into an indefinite inner product $\inntw\defeq \langle\, \cdot \, , \, \tw \cdot \, \rangle$ with the unitary operator $\tw=\tw^\dagger$. Such structures are essential in quantum field theory to describe spinor fields on pseudo-Riemannian manifolds. However, developing a fully satisfactory and widely accepted framework for Lorentzian noncommutative geometry remains an open and actively pursued area of research.

	In \cite{nieuvsignchange}, I present a connection between twisted spectral triples and pseudo-Riemannian spectral triples based on a fundamental link between the previous Krein product and the twist defined by $\rho(\, \cdot\,)\defeq \tw (\, \cdot\,)\tw$. A concept of spectral triple morphism was introduced, enabling the transformation of a twisted spectral triple into a pseudo-Riemannian one. This morphism was shown to induce a signature change in the context of the manifold's spectral triples. The signature change transformation was determined solely by the unitary operator $\tw$, being the central element of the approach.
	
	The following article proposes to extend these results. The first three sections introduce what are twisted spectral triples, pseudo-Riemannian spectral triples, and the $\tw$-morphism between them. One of the main contributions is the definition of a specific real structure for twisted spectral triples in subsection \ref{SecTwistedRealStruct}, adapted to the bimodule structure associated with the derivations and related to the special $\tw$-unitary operators which generates the twisted fluctuations. Interestingly, this real structure also emerges in the setting of pseudo-Riemannian spectral triples and remains invariant under the $\tw$-morphism, thus retaining its significance in both frameworks.
	
	Another important point was the introduction of the twist in the context of Clifford algebras (for even-dimensional vector spaces) in subsection \ref{SubsecCliff}, showing its structural role in the structure, to characterise the orthochronous spin group and to deduce the Krein product $\inntw$. The way the $\tw$-morphism acts to change the signature of the metric and the Christoffel symbols is presented in section \ref{subsectionsignaturechange}.
	
	The main results are presented in section \ref{sec Almost Commutative manifold}. An almost-commutative spectral triple is constructed on a twisted spectral triple, naturally inducing the emergence of the operator $\tw$ in the product Dirac operator. This leads to the appearance of a Krein structure together with the pseudo-Riemannian Dirac operators, from the very first spectral triple structure, then implementing the $\tw$-morphism. We then show how the specific signature of space-time is recovered in the 4-dimensional case. This section follows a top-down approach, revealing the signature change as an intrinsic outcome of the algebraic structure governing the almost-commutative spectral triple.

	\section{Twisted spectral triple}
	\label{SecTwisted}
	
	Spectral triples have been particularly successful in describing noncommutative spaces arising from type I and type II von Neumann algebras, where a well-defined trace exists. However, a major challenge in noncommutative geometry has been the extension of these methods to type III von Neumann algebras, which lack such a finite normal trace. Type III algebras arise naturally in quantum field theory, especially in the algebraic approach to quantum statistical mechanics and the study of vacuum states in curved spacetime.
	
	To address this challenge, twisted spectral triples were introduced by A. Connes and H. Moscovici in their 2006 work, see \cite{connes2006type}. These modified spectral triples generalize the standard construction by incorporating an automorphism, which compensates for the absence of a trace and allows for the formulation of geometric invariants in this setting.
	
	Interestingly, the introduction of this twist is related to the modular theory of Tomita-Takesaki, which governs the structure of von Neumann algebras. Tomita-Takesaki theory links the modular structure of von Neumann algebras to time evolution, defining a natural flow on the algebra. Twisted spectral triples were then implemented in the context of the noncommutative Standard Model in \cite{TwistSpontBreakDevastaMartine2017,TwistLandiMarti2016,TwistGaugeLandiMarti2018}. Applications to conformal and noncommutative geometry can also be found in \cite{ponge2016index, ponge2018noncommutative}.

	\subsection{Fundamentals of twisted spectral triple}
	
	Let $\calA$ be a $*$-algebra with representation $\pi$ on a Hilbert space $\calH$. The faithful representation $\pi$ will be omitted by identifying the elements of $\calA$ with their representation. The adjoint $\dagger$ in $\calB(\calH)$ is then given by the relation $\pi(a^*)=\pi(a)^\dagger$.
	
	\begin{definition}[Twisted spectral triple]
		\label{DefTwisted}
		A twisted spectral triple $(\calA, \calH, \Dir, \rho)$ is obtained by considering an involutive algebra $\calA$ acting faithfully as a bounded operator algebra on a Hilbert space $\calH$ with a self-adjoint operator $\Dir$ with compact resolvent together with an automorphism $\rho\in \Aut(\calA)$ called twist which satisfies the regularity condition
		\begin{equation}
			\label{EqRegRelation}
			\rho(a^\dagger)=(\rho^{-1}(a))^\dagger,\,\qquad\qquad \forall a\in\calA.
		\end{equation}
		The twisted commutator
		\begin{equation}
			[\Dir,a]_\rho:= \Dir a-\rho(a)\Dir
		\end{equation}
		is used instead of $[\Dir,a]$ and we require only the twisted commutator to be bounded.
	\end{definition} 
	%The derivation is now given by it's twisted version $\delta_\rho (a):= [\Dir,a]_\rho$.
	
	In the same way as for spectral triples, we can add a grading $\Gamma$ and a real structure $J$. The axioms of such even real twisted spectral triples are the same as for spectral triples
	\begin{equation}
		\label{EqDefStructParam}
		J^2 = \epsilon_0,\qquad J D = \epsilon_1 D J,\qquad J \Gamma = \epsilon_2 \Gamma J,\qquad \epsilon_0, \epsilon_1, \epsilon_2=\pm 1,
	\end{equation}
	with the introduction of the characteristic number $\epsilon=\pm 1$ so that $\rho$ is required to satisfy
	\begin{equation}
		\label{EqRelKJ}
		\rho( J)=\epsilon J,
	\end{equation}
	and the replacement of the first order condition by the twisted first order condition
	\begin{equation}
		\label{TwFirstorder}
		[[\Dir, a]_\rho, b^\circ]_{\rho^\circ}=0 
	\end{equation}
	$\forall a\in\calA$ and $b^\circ \in\calA^\circ$ with $\calA^\circ$ the opposite algebra. The opposite twist $\rho^\circ$ acting on $\calA^\circ$ is defined by 
	\begin{equation}
		\rho^\circ(a^\circ):= (\rho^{-1}(a))^\circ.
	\end{equation}
	See \cite{TwistLandiMarti2016} for further details.
	
	The set of twisted 1-forms $\Omega^1_\Dir(\calA, \rho)$ is defined as 
	\begin{equation}
		\Omega^1_\Dir(\calA, \rho)\defeq \{\sum_ka_k[\Dir, b_k]_{\rho}\, \mid\,   a_k, b_k \in \calA\}
	\end{equation}
	so that twisted-fluctuations of $\Dir$ are given by
	\begin{equation}
		\label{EqTwFluct}
		\Dir_{A_\rho}:= \Dir+A_\rho+\epsilon_1 JA_\rho J^{-1}\qquad\text{with}\qquad A_\rho\in \Omega^1_\Dir(\calA, \rho).
	\end{equation}
	The set $\Omega^1_\Dir(\calA, \rho)$ is a $\calA$-bimodule for the following product
	\begin{equation}
		\label{EqLeftRightTw}
		a\, \cdot\, A_\rho \, \cdot\,  b=\rho(a)A_\rho b
	\end{equation}
	for any $a, b\in\calA$ and $A_\rho \in \Omega^1_\Dir(\calA, \rho)$. This makes the map $a\,\to\, \delta_\rho(a)\defeq [\Dir,a]_\rho$ be a derivation since we will have that
	\begin{equation}
		\delta_\rho(ab)=\delta_\rho(a)b+\rho(a)\delta_\rho(b)=\delta_\rho(a)\, \cdot\, b+a \, \cdot\,  \delta_\rho(b).
	\end{equation} 
	As shown in \cite{TwistLandiMarti2016} a way to  “minimally twist” the spectral triple $(\algA, \calH, \Dir, \Gamma)$ i.e.,without changing the Hilbert space and the Dirac operator, is by the use of the decomposition coming with $\Gamma$. Using the projections $p_\pm:=  (\bbbone \pm \Gamma)/2$ so that $\calH=p_+\calH\oplus p_-\calH:=\calH_+\oplus \calH_-$, all elements $a=(a_1,a_2)\in \algA\otimes \mathbb{C}^2$ are represented accordingly:
	\begin{equation}
		\label{EqRhogen}
		\pi(a_1,a_2)\defeq p_+\pi_0(a_1)+p_-\pi_0(a_2)=	\begin{pmatrix}
			\pi_+(a_1) & 0  \\
			0 & \pi_-(a_2)  \\
		\end{pmatrix}
	\end{equation}
	with $\pi_0$ the representation used for $\algA$ and $\pi_\pm(\, \cdot\, )\defeq p_\pm \pi_0(\, \cdot\, )_{|\calH_\pm}$ the restrictions to $\calH_\pm$. This will be called the chiral representation.
	\begin{definition}[twist by grading]
		Taking the graded spectral triple $(\algA, \calH, \Dir, \Gamma)$, its corresponding twist by gradding is given by the twisted spectral triple $(\algA\otimes \mathbb{C}^2 , \calH, \Dir, \Gamma, \rho )$, using the chiral representation. The twist acts on elements $a=(a_1,a_2)\in \algA\otimes \mathbb{C}^2$ as
		\begin{equation}
			\label{EqactionRho}
			\rho(a_1,a_2)= (a_2,a_1).
		\end{equation}  
	\end{definition}  
	As shown in \cite{TwistLandiMarti2016}, the twist by the grading is the only way to minimally twist the spectral triple of even-dimensional manifolds based on $\CM\otimes \mathbb{C}^2$. The fact that no such procedure for twisting odd-dimensional manifold's spectral triple have been found will be one of the justifications to focus on the study of even-dimensional manifolds. More information on twisted spectral triples can be found in \cite{TwistLandiMarti2016,TwistGaugeLandiMarti2018, martinetti2024torsion}.
	
	\newpage
	\subsection{From twisted to indefinite inner products}
	\label{subsecFromTwToIndef}
	
	Following \cite{nieuvsignchange}, any automorphism $\rho\in \Aut(\calA)$ will be required to be $\calB(\calH)$-regular, i.e., a regular automorphism on all of $\calB(\calH)$. Since \cite{devastato2018lorentz}, a novelty consists in the possibility of considering a new inner product $\langle\, \cdot \, , \,\cdot  \,   \rangle_\rho$, called the $\rho$-product, which by definition satisfies
	\begin{align}
		\label{EqRhoprod}
		\langle\psi   , O\psi^\prime   \rangle_\rho=\langle\rho(O)^\dagger\psi ,\psi^{\prime} \rangle_\rho
	\end{align}
	for all $\psi, \psi^{\prime} \in \calH$ and $O\in \calB(\calH)$.
	
	This new inner product naturally comes with its own notion of adjoint, denoted by "$+$", defined as $(\, \cdot\,)^+:= \rho(\, \cdot\, )^\dagger$, along with an associated notion of $\rho$-unitarity, characterized by elements $U_\rho$ satisfying $U_\rho U_\rho^+=U_\rho^+U_\rho=\bbbone$, thereby preserving the $\rho$-product. As mentioned in \cite{nieuvsignchange}, the uniqueness of this adjoint (i.e., the equality of the left and right adjoint) is ensured by the $\calB(\calH)$-regularity condition. We will restrict ourselves to twists that are $*$-automorphisms, i.e., $\rho(O^\dagger)=\rho(O)^\dagger$. This, together with the regularity condition, implies that $\rho^2=\bbbone$, which in turn ensures that $(\, \cdot  \,)^+$ is an involution.
	
	As shown in \cite{TwistGaugeLandiMarti2018} and \cite{martinetti2024torsion}, there are two ways to generate the twisted fluctuations of equation \eqref{EqTwFluct}, using either usual unitary operators or $\rho$-unitary operators:
	\begin{enumerate}
		\item Type 1: $\Dir_{A_\rho}=U\Dir U^+$ with $U=uJuJ^{-1}$ and $u$ a unitary operator in $\calA$.
		\item Type 2: $\Dir_{A_\rho}=U_\rho\Dir U_\rho^\dagger$ with $U_\rho=u_\rho Ju_\rho J^{-1}$ and $u_\rho$ a $\rho$-unitary operator in $\calA$.
	\end{enumerate}
	This justifies the introduction of the notions of $\rho$-product, along with the adjoint "$+$" and the notion of $\rho$-unitarity.
	
	In the case where $\rho\in\Inn(\calB(\calH))$, there exists a unitary operator $\tw$ acting on $\calH$ such that $\rho(O)=\tw O \tw^\dagger$. In this case, the $\rho$-product is determined by $\tw$, so that $\langle\, \cdot\, , \, \cdot\, \rangle_\rho := \langle\, \cdot\, , \tw \cdot  \, \rangle$. Furthermore, the twist is a $*$-automorphism of $\calA$.
	
	In the following, any twist $\rho$ will be assumed to be in $\Inn(\calB(\calH))$. Since different operators $\tw$ can implement the same twist $\rho$ while inducing different inner products, we will adopt a more precise notation. The twisted spectral triple will be denoted by $(\calA, \calH, \Dir, \tw)$, and the $\rho$-product notation $\langle\, \cdot\, , \, \cdot\, \rangle_\rho$ will be replaced by the $\tw$-product notation:
	\begin{equation}
		\label{Eqtwprod}
		\inntw := \langle\, \cdot\, , \tw \cdot \, \rangle
	\end{equation}
	
	so that the associated adjoint is called the $\tw$-adjoint and is defined by
	\begin{equation}
		(\, \cdot\,)^\Tadj\defeq \rho(\, \cdot\,)^\dagger.
	\end{equation}
	We define $\tw$-unitary operators by the relation $\utw \utw^\Tadj =\utw^\Tadj  \utw=\bbbone$, with the corresponding set denoted as $\calU_\tw(\calB(\calH))$.
	
	Taking the automorphism $\rho(\, \cdot\,)=\tw(\, \cdot\,)\tw^\dagger$ to be $\calB(\calH)$-regular, or regular on $\calA$ if the algebra is a simple finite-dimensional algebra, we obtain
	\begin{equation}
		\label{EqFormTw}
		\tw=\exp(i\theta)\tw^\dagger
	\end{equation}
	for some real number $\theta$. As a consequence, $\inntw$ is Hermitian and indefinite if and only if $\tw=\tw^\dagger$ (see \cite{nieuvsignchange} for more details). The condition $\tw=\tw^\dagger$ will be derived from algebraic constraints in the top-down approach presented in section \ref{sec Almost Commutative manifold}.

	\begin{definition}[Fundamental twist]
		A fundamental twist is a twist $\rho\in\Inn(\calB(\calH))$, implemented by the unitary $\tw$ such that $\tw=\tw^\dagger$. 
	\end{definition}

	\subsection{Twisted real structure}
	\label{SecTwistedRealStruct}
	
	In noncommutative geometry, the real structure given by the operator $J$ allows to define a right action of the algebra 
	$\calA$ on the Hilbert space $\calH$, turning it into a $\calA$-bimodule by relating left and right actions. We have seen in equation \eqref{EqLeftRightTw} that $\Omega^1_\Dir(\calA, \rho)$ is a $\calA$-bimodule whose left and right actions are connected by the twist $\rho$. Here, I propose to introduce a notion of real structure that implement the twisted bimodule action. This particular real structure will play a crucial role throughout the article, appearing notably in the context of Clifford algebras with indefinite (pseudo-Riemannian) metrics in subsection \ref{SubsecCliff}, and then in the framework of pseudo-Riemannian spectral triples in subsection \ref{SubsecIndefinite}. 
	
	Taking $a\in\calA$, the left action of the algebra on bimodules is denoted by $l(a)$. This operator satisfies $l(a^\dagger)=l(a)^\dagger$. The corresponding right action $r(a)$ is given through $J$ by 
	\begin{align}
		r(a)=Jl(a)^\dagger J^{-1},
	\end{align}
	with the zero-order condition $[l(a), r(a)]=0$, ensuring that the right and left actions on the $\calA$-bimodules are connected by the relation
	\begin{equation}
		a \cdot (\, \cdot\, ) \cdot  b=l(a)r(b)(\, \cdot\, ).
	\end{equation}
	
	For the usual action on a bimodule, we have $a\cdot(\, \cdot\, )\cdot  b=a(\, \cdot\, )b$, meaning that the left and right actions correspond to the usual multiplication of operators acting on elements of the bimodule. We denote by $\JH$ the corresponding real structure, i.e., $\hat{l}(a)=a$ and $\hat{r}(a)=\JH \hat{l}(a)^\dagger \JH^{-1}$. We define $\hat{r}(a)=a^{\hat{\circ}}=\JH a^\dagger \JH^{-1}$ as the element of the corresponding opposite algebra $\calA^{\hat{\circ}}$. This is the real structure used for conventional spectral triples.
	
	Now, to recover the twisted product action \eqref{EqLeftRightTw}, i.e.,
	\begin{equation}
		a\cdot(\, \cdot\, ) \cdot b=\rho(a)(\, \cdot\, ) b,
	\end{equation}
	we have to find $J$ such that $l(a)=Jr(a)^\dagger J^{-1}=\rho(a)$ with $r(a)=\hat{r}(a)=a$. Noting that
	\begin{align}
		l(a)=\rho(\hat{l}(a))=\rho(\JH \hat{r}(a)^\dagger \JH^{-1})=\rho(\JH r(a)^\dagger \JH^{-1}),
	\end{align}
	we can define\footnote{In full generality, we may have $J\defeq \alpha\tw \JH$ with $\alpha\in\mathbb{C}$ such that $\alpha\alpha^*=1$, but this phase shift is set to $1$ to avoid unnecessary terms in the expressions.}
	\begin{equation}
		\label{EqRelJ}
		J\defeq \tw \JH
	\end{equation}
	so that $l(a)=Jr(a)^\dagger J^{-1}=\rho(a)$. We define $r(a)=a^{\circ}=\rho(a^{\hat{\circ}})=J a^\dagger J^{-1}$ as the element of the corresponding opposite algebra $\calA^{\circ}$.
	
	In subsection \ref{subsecFromTwToIndef}, we highlight the important role the $\tw$-unitary operators associated with $\inntw$ play in the context of twisted spectral triples. Given any $(\psi, u_\tw) \in (\calH, \calU_\tw(\calA)$, then the adjoint action on the bimodule is given by
	\begin{equation}
		\label{EqAdjActionUni}
		\psi\,\to\, u_\tw \psi u_\tw^{-1}=u_\tw \psi u_\tw^\Tadj=u_\tw J u_\tw J^{-1}\psi 
	\end{equation}
	showing that the adjoint action of $\tw$-unitary operators is intrinsically connected with the operator $J$ defining the twisted real structure.
	
	Spectral triples correspond to the special case $\tw=\bbbone$. The twisted commutator then reduces to the usual one, $u_\tw=u$ is unitary for $\inn$, and $J=\JH$, so that we recover the usual adjoint action on the bimodule:
	\begin{equation}
		\psi\,\to\, u \psi u^{-1}=u \psi u^\dagger=u \JH u \JH^{-1}\psi. 
	\end{equation}
	
	\subsection{Twisted Clifford relation}
	\label{SubsectwistCliffRel}
A general definition of the notion of twisted Clifford algebra was given in \cite{nieuvsignchange}. However, this was not clear weather the twisted Clifford algebra can be represented from the defining twisted Clifford algebraic constrait of the vector space and the metric. We propose here a modification of this definition to solve this issue, by retrieving the well known Clifford algebra case, but from a twisted Clifford algebraic relation. This specific relationship to the metric will be connected to a twisted spectral triple in subsection \ref{subsectionsignaturechange}.

\begin{definition}[Twisted Clifford relation]
	Let $V$ be a vector space over a commutative field with $\tg$ a metric on $V$ and $\rho$ a linear, unit-preserving involution on $V$, i.e., $\rho^2 = id_V$. The relation between $\tg$ and $V$ is given by the twisted Clifford relation
	\begin{align}
		\label{EqDefTwCliff}
		\rho(u)\cdot v + v\cdot \rho(u) = 2 \tg(u,v)\bbbone,\qquad\qquad \forall u, v\in V
	\end{align}  
	where $\rho$ is required to be a morphism, i.e., $\rho(u\cdot v) = \rho(u)\cdot \rho(v)$ for any $u,v \in V$.
\end{definition}

The bilinearity of $\tg$ follows from the linearity of $\rho$, and the symmetry of $\tg$ from the fact that $\rho$ is a unital, $V$-involutory morphism.

\begin{definition}[Reflection]
	\label{DefReflection}
	Taking a vector space $V$ and the metric $\tg$ on $V$, a reflection is a $\tg$-isometric involutive map on $V$.
\end{definition} 

\begin{proposition}
	$\rho$ is a reflection for $\tg$.
\end{proposition}

\begin{proof}
	We only have to prove that $\rho$ is $\tg$-isometric. The computation of the action of $\rho$ on equation \eqref{EqDefTwCliff} gives
	\begin{align*}
		u\cdot \rho(v) + \rho(v)\cdot u = 2 \tg(u,v)\bbbone
	\end{align*}
	where the left-hand side also identifies with $2 \tg(\rho(u),\rho(v))\bbbone$, implying that $\tg(u,v) = \tg(\rho(u),\rho(v))$, which proves that $\rho$ is $\tg$-isometric.
\end{proof}

We define $\g$ from $\tg$ and $\rho$ by the relation
\begin{align}
	\g(u,v) = \tg(\rho(u),v),\qquad\qquad \forall u,v \in V.
\end{align}

The bilinearity of $\g$ follows from that of $\rho$ and $\tg$, and the symmetry property from
\begin{align*}
	\g(u,v) = \tg(\rho(u),v) = \tg(v, \rho(u)) = \tg(\rho(v), \rho^2(u)) = \tg(\rho(v), u) = \g(v,u).
\end{align*}
 
Then equation \eqref{EqDefTwCliff} rewrites
\begin{align*}
	\rho(u)\cdot v + v\cdot \rho(u) = 2 \g(\rho(u),v)\bbbone 
\end{align*}
or equivalently, replacing $\rho(u)$ by $u$\footnote{This is possible since $\rho$ is an involution, hence an element in $\Aut(V)$.} 
\begin{align}
	\label{EqDefCliff}
	u\cdot v + v\cdot u = 2 \g(u,v)\bbbone  
\end{align}

The algebra generated by $\bbbone$ and $V$ is then, by definition, the Clifford algebra $Cl(V, \g)$. Automorphisms $\Aut(Cl(V, \g))$ of the Clifford algebra are bijective maps $\phi\, :\, Cl(V, \g)\, \to\, Cl(V, \g)$ that preserve equation \eqref{EqDefCliff}.

\begin{proposition}
	The involution $\rho$ is an element of $\Aut(Cl(V, \g))$.
\end{proposition}

\begin{proof}
	Any element in $Cl(V, \g)$ is a linear combination of $\bbbone$ and elements of the form $\alpha v_1v_2\dots v_k$ with $v_i\in V$, $k\in\mathbb{N}$ and $\alpha\in \mathbb{C}$. The action of $\rho$ on such elements is given by $\rho(\bbbone)=\bbbone$ and $\rho(\alpha v_1v_2\dots v_k)=\alpha \rho(v_1)\rho(v_2)\dots \rho(v_k)$, which lies in $Cl(V, \g)$ as $\rho(v)\in V$ for any $v\in V$, proving that $\rho\, :\, Cl(V, \g)\, \to\, Cl(V, \g)$. It remains to prove that $\rho$ preserves equation \eqref{EqDefCliff}. From a direct computation we have
	\begin{align}
		\label{EqAutomPr}
		\rho(u)\cdot \rho(v) + \rho(v)\cdot \rho(u) = 2 \g(u,v)\bbbone
	\end{align}
	proving that $\rho\in \Aut(Cl(V, \g))$.
\end{proof}

Note that $\rho$ is also a reflection for $\g$ since equation \eqref{EqAutomPr} identifies with $2 \g(\rho(u),\rho(v))\bbbone$.

If $\rho$ is implemented by an operator $\tw$ such that $\rho(v)=\tw\cdot v\cdot\tw$ with $\tw^2=\tw\cdot\tw=\bbbone$, we can define $\tilde{v}\defeq \tw \cdot v$ so that equation \eqref{EqDefTwCliff} can be rewritten as
\begin{align}
	\label{DefRelDeuxTwRel}
\tilde{u}\cdot \tilde{v}+\rho(\tilde{v}\cdot \tilde{u})=2 \tg(u,v)\bbbone.
\end{align}
Such a reformulation of the metric relation in terms of the vectors $\tilde{v}$ will be particularly relevant in subsections \ref{subsectionsignaturechange} and \ref{Subsec4dimACmfld}.

The definitions of this subsection do not imply any regularity condition satisfied by $\rho$. The following section will provide a context in which the twist $\rho$ satisfies the $\mathcal{B}(\mathcal{H})$-regularity condition for a given Hilbert space $\calH$.

\section{Pseudo-Riemannian spectral triples}
	\label{sec Lorentz symetry}
	\subsection{Krein product and pseudo-Riemannian spectral triple}
	\label{SubsecIndefinite}
	
In the previous section, a notion of $\tw$-product was introduced in the context of twisted spectral triples. This notion coincides with the Krein product when $\tw = \tw^\dagger$. We now propose to present this notion of Krein product, along with the associated concept of Krein space, leading to the definition of pseudo-Riemannian spectral triples.
	
Let us start with a Hilbert space $\calH\defeq (E, \langle \, \cdot\, , \, \cdot\, \rangle)$ where $E$ is a vector space and $\inn$ a Hermitian inner product. Taking a self-adjoint unitary operator $\tw^\dagger =\tw$ (called fundamental symmetry), the associated Krein product is given by
	\begin{equation}
		\label{EqKreinProd}
		\inntw\defeq \langle \, \cdot\, , \tw\, \cdot\, \rangle.
	\end{equation}
The Krein space is defined as
	\begin{align*}
		\calK\defeq (E, \inntw).
	\end{align*}

	Considering an involutive unital algebra $\algA$ represented by bounded operators on $\calH$, the operator $\tw$ should not be regarded as an element of $\algA$. Consequently, the corresponding $\tw$-unitary actions on the algebra and the Hilbert space do not act on it:
	\begin{equation*}
		\langle \psi , a\psi \rangle_\tw \,\to\, \langle \utw\psi , \utw a{\utw}^\Tadj \utw \psi \rangle_\tw= \langle \psi , \utw^\Tadj \utw a\utw^\Tadj \utw \psi \rangle_\tw=\langle \psi , a\psi \rangle_\tw,
	\end{equation*}
	where $\utw\in \calU_\tw(\calB(\calH))$. This justifies the notation $\inntw$.

	Conversely, fixing $\tw$ as part of the structure constrains the transformations $\psi\,\to\, U\psi$ and $\psi\,\to\, U^\dagger\psi$ that preserve $\langle \psi , \psi \rangle_\tw$ to be $\tw$-unitaries, since they must satisfy
	\begin{equation}
		U^\dagger \tw U=U \tw U^\dagger=\tw,
	\end{equation}
	which implies that 
	\[
	UU^\Tadj=U^\Tadj U=\bbbone.
	\]

	We can then generalize the notion of spectral triples by defining a spectral triple based on these Krein space structures.
	
	\begin{definition}[Pseudo-Riemannian spectral triple]
		A pseudo-Riemannian spectral triple $(\calA,\calK, \KDir)$ is a spectral triple in which the Hilbert space is replaced by a Krein space $\calK$, and for which the Dirac operator $\KDir$ is $\Tadj$-self-adjoint, with the requirement that $[\KDir ,a]$ be bounded for any $a\in\calA$. 
	\end{definition} 
	
	If the pseudo-Riemannian spectral triple admits a real structure, an operator $J$ that implements the adjoint action of $\tw$-unitary operators in $\calA$ can be derived from $\JH$ (defined in Subsection \ref{SecTwistedRealStruct}) and $\tw$ via the relation \eqref{EqRelJ}.

	The algebraic constraints of a pseudo-Riemannian spectral triple are specified by:
	\begin{equation}
		\label{EqDefStructParamK}
		J^2 = \epsilon_0^\tw, \qquad\quad	 J \KDir = \epsilon_1^\tw \KDir J, \qquad\quad	 J \Gamma = \epsilon_2^\tw \Gamma J, \qquad\quad	 \epsilon_0^\tw, \epsilon_1^\tw, \epsilon_2^\tw \in \{\pm 1\}.
	\end{equation}
	The following table gives the relation between the KO-dimension of even-dimensional pseudo-Riemannian spectral triples and the numbers $\epsilon_i^\tw$ for $i\in\{0, 1, 2\}$:
	
	\begin{table}[H]
		\centering
		\small
		\resizebox{0.4\textwidth}{!}{%
			\begin{tabular}{|c|*{4}{c|}}
				\hline
				$2(n-m) \bmod 8$	& 0 & 2 & 4 & 6\\
				\hline
				$\epsilon_0^\tw$ & 1 & -1 & -1 & 1 \\ 
				\hline
				$\epsilon_1^\tw$ & 1 & 1 & 1 & 1  \\ 
				\hline
				$\epsilon_2^\tw$ & 1 & -1 & 1 & -1 \\ 
				\hline
		\end{tabular}}
		\normalsize
		\caption{\label{KoDimTable} KO-dimensions for even-dimensional pseudo-Riemannian spectral triples.}
	\end{table} 
	
	As for twisted spectral triples, standard spectral triples are retrieved in the special case $\tw=\bbbone$, which corresponds to $n=0$ in the KO-dimension table. More information on such spectral triples can be found in \cite{strohmaier2006noncommutative}.
	
The following three sections aim to construct the pseudo-Riemannian spectral triple on even-dimensional manifolds.

	\subsection{Clifford algebras for even-dimensional vector spaces and associated twist}
	\label{SubsecCliff}
	
	From now on, we will focus on Clifford algebras generated by complex linear representations of even-dimensional real vector spaces, in the spirit of quantum field theory with gamma matrices. The aim is to present the pseudo-Riemannian spectral triple of a manifold, and the connection with the twist.
	
	The following subsections aims at introducing such Clifford algebras with the so-called three structural automorphisms, their properties and interrelations. These automorphisms are of great importance to characterize these Clifford algebras and their automorphisms. The main novelty consists in the introduction of the twist (satisfying the $\mathcal{B}(\mathcal{H})$-regularity condition; see Proposition \ref{PropRegKK}), being one of these automorphisms, connected to the metric and providing the parity operator in the case the chosen basis is constituted of unitary operators. This twist will also plays a significant role in subsection \ref{subsecAutomClifftoKrein}, to characterize an important class of automorphisms, which holds fundamental importance in physics, i.e., the orthochronous spin group and the associated indefinite inner product. The author has followed similar conventions to \cite{lawson2016spin}.

	Let $V\simeq \mathbb{R}^{2m}$ be an even-dimensional (real) vector space over $\mathbb{R}$ of dimension $2m$, with $m\in\mathbb{N}$. A non-degenerate symmetric bilinear form $\g$ on $V$ determines a metric with signature $(n, 2m-n)$, where $n\in\mathbb{N}$. 
	
	The automorphism group of $V$ is the general linear group $\glv$ implemented by $2m$-square matrices $A$ with real coeficients acting on $V$. The orthogonal group $\ovg$ is the subgroup of $\glv$ composed by elements that preserves the metric, i.e.,$A\in \glv$ such that $\g(Au, Av)=\g(u,v)$.

	Taking any vector $v\in V$, a complex linear representation of $v$ is provided by the Clifford multiplication operator $\cl(v)$ as a $2^m \times 2^m$ complex square matrix such that
	\begin{equation}
		\label{EqDefRelCliff}
		\cl(u)\cl(v) + \cl(v)\cl(u) = 2\g(u, v)\mbb,\qquad \qquad \forall\,  u, v \in V.
	\end{equation}
The product $\cl(u)\cdot\cl(v)=\cl(u)\cl(v)$ is provided by the operator multiplication here.

The definition of $\glv$ extends to the Clifford representation of $V$ by the relation  
	\begin{align}
		\label{RelExtent}
	 \phi_A(\cl(v))\defeq\cl(Av),\qquad\qquad\forall A\in \glv\,\,\, \text{and}\,\,\, v\in V,
	\end{align}

Unless explicitly mentioned, we will identify elements $v\in V$ with their representation $\cl(v)$ in what follows, thereby identifying the vector space $V$ with its representation and also $\g(u, v)$ with $\g(\cl(u), \cl(v))$.
 
We denote by $\glpv$ the elements of $\phi_A\in \glv$ which are also unital morphisms of $V$ equiped with the product "$\cdot$" i.e.,$\phi_A(\cl(u)\cl(v))=\phi_A(\cl(u))\phi_A(\cl(v))$ forall $u,v\in V$. The notation $\opvg$ will be used for the elements of $\ovg$ in $\glpv$. We have that $\glpv=\opvg$ since the action of $\glpv$ preserve relation \eqref{EqDefRelCliff}. Note that $\rho$ defined in subsection \ref{SubsectwistCliffRel} is an element in $\opvg$.

	The (complex) Clifford algebra $Cl(V, \g)$ is the algebra generated by the vector space $V\subset Cl(V, \g)$ and the identity $\mbb\subset Cl(V, \g)$.

 Since $Cl(V, \g)$ has even dimension, it is isomorphic to $M_{2^m}(\mathbb{C})$, implying that all its automorphisms are inner, i.e., $\Aut(Cl(V, \g))=\Inn(Cl(V, \g))$. Defining $Cl^{\, \times}(V, \g)$ as the subset of invertible elements in $Cl(V, \g)$,  then every automorphism of $Cl(V, \g)$ can be characterized by an element $x\in Cl^{\, \times}(V, \g)$ through the adjoint action
	\begin{equation}
		\phi_x(y)\defeq\ad_x(y)=xyx^{-1},\qquad \qquad\qquad \forall y\in Cl(V, \g).
	\end{equation}
 
	Note that $\opvg\subset \Aut(Cl(V, \g))$ can then be implemented by such an adjoint action.

	Let $\{e_a\}_{1\leq a\leq 2m}$ be a $\g$-orthonormal basis of $V$, i.e., subject to the relation
	\begin{equation}
		e_ae_b+e_be_a=2\g(e_a,e_b)\mbb =\left\{
		\begin{array}{ll}
			\,\,\,\,2\delta_{ab}\mbb  &\qquad\text{ for $1\leq a\leq n$},\\ 
			-2\delta_{ab}\mbb  &\qquad\text{ for $n+1\leq a\leq 2m$}. 
		\end{array}\right.  
	\end{equation}
	Any vector $v\in V$ can be written as $v=\sum_{a=1}^{2m} \lambda_a e_a$ with $\lambda_a\in \mathbb{R}$ (since $V$ is a linear representation of $\mathbb{R}^{2m}$). $Cl(V, \g)$ is then generated by $\{e_a\}_{1\leq a\leq 2m}$ and $\mbb$, so that a basis of $Cl(V, \g)\simeq M_{2^m}(\mathbb{C})$ is given by the set
	\begin{equation}
		\mathcal{I}\defeq \{\mbb, \{e_{a} \}_{1\leq a\leq 2m},\{e_{a_1}e_{a_2} \}_{1\leq a_1<a_2\leq 2m},\{e_{a_1}e_{a_2}e_{a_3} \}_{1\leq a_1<a_2<a_3\leq 2m},\dots ,e_{1}e_{2}\dots e_{2m}  \}.
	\end{equation}
	Let define $\g_v\defeq \g(v,v)$ and $\g_a\defeq \g(e_a,e_a)$. The metric $\g$ induces the splitting $V=V^+\oplus V^-$, where the set $\{e_a\}_{1\leq a\leq n}$ is a basis of $V^+$ and the set $\{e_a\}_{n+1\leq a\leq 2m}$ is a basis of $V^-$, with $\g_{v^+}\geq 0$ and $\g_{v^-}\leq 0$ for $v^\pm \in V^\pm$.
	
	An irreducible representation of $Cl(V, \g)\simeq M_{2^m}(\mathbb{C})$ is given by the module $M=\mathbb{C}^{2^m}$. A Hermitian inner product $\inn$ can be associated with $M$, with adjoint denoted by $\dagger$. We define the Hilbert space $\calH\defeq(M,\langle \, \cdot\, , \, \cdot\, \rangle)$.
	%  The basis $\{e_a\}_{1\leq a\leq 2m}$ is assumed to be orthonormal with respect to $\langle \, \cdot\, , \, \cdot\, \rangle$.\GN{enlève peut etre} 
	
	The commutant $V^\prime$ of $V$ in $Cl(V, \g)$ is defined by
	\begin{align}
		V^\prime\defeq \{a\in Cl(V, \g) \, \mid\, \forall v\in V, \quad va=av\}.
	\end{align}
	\begin{proposition}
		\label{PropComV}
		We have that $V^\prime=Z(Cl(V, \g) )=\mathbb{C}\mbb$.
	\end{proposition}
	\begin{proof}
		Since $Cl(V, \g)\simeq M_{2^m}(\mathbb{C})$, its center is $\mathbb{C}\mbb$. Now, taking any element $a\in Cl(V, \g)$, it is a linear combination of elements in $\mathcal{I}$. The requirement $[a,v]=0$ for any $v\in V$ implies that $a\in \mathbb{C}\mbb$, since for any element $e_{a_1}e_{a_2}\dots e_{a_k}\in \mathcal{I}$ with $1\leq a_1<a_2<\dots <a_k\leq 2m$, there exists an element $e_a\in V$ that anticommutes with it, meaning that $V^\prime=\mathbb{C}\mbb$.
	\end{proof}
	
	The complex conjugation operator is defined on any operator $O\in\mathcal{B}(\mathcal{H})$ by $\overline{O}=cc\circ O\circ cc$ where $cc=cc^{-1}$ is the complex conjugation over $\mathcal{H}$.
	
	%  	The structural (linear) automorphisms $\rho, \chi, \kappa\in \Aut(Cl(V, \g))$, called twist, grading, and charge conjugation, are defined for all $v\in V$ by 
	\begin{definition}
		\label{DefKCliff}
		The structural operators $\rho, \chi, \kappa$, called twist, grading, and charge conjugation, are defined to be elements in $\opvg$ so that $\forall v\in V$ we have
		\begin{align}
			\label{EqDefStruct}
			&\rho(v)\defeq v^\dagger, \qquad\qquad\quad\chi(v)\defeq -v, \qquad\qquad\quad\kappa(v)\defeq -\bar{v}.
		\end{align}
	\end{definition}
	This directly implies that $\rho, \chi, \kappa$ are reflections on $V$, that $\bar{v}, v^\dagger$, and $v^\top=\bar{v}^\dagger$ are elements of $V$ for any $v\in V$, and consequently that 
	\begin{equation}
		\rho(\bar{v})=v^\top =\overline{\rho(v)}.
	\end{equation}
	We then have the equivalent definition $\kappa(v)=\chi(\bar{v})$.
	\begin{remark}
		If we do not restrict to even-dimensional vector spaces $V$, the relation for $\kappa$ is more generally given by
		\begin{equation}
			\kappa(v)\defeq -\epsilon_1^\tw\bar{v},
		\end{equation}
		where $\epsilon_1^\tw=\pm 1$. As shown in the KO-dimension table \ref{KoDimTable}, a particular feature of the even-dimensional case is that $\epsilon_1^\tw=1$ in all cases, so the definition without this factor.
	\end{remark}

	\begin{proposition}
		\label{PropInvol}
		The structural automorphisms $\rho,\chi$, and $\kappa$ are involutions on $Cl(V, \g)$.
	\end{proposition}
	\begin{proof}
		Any element of $Cl(V, \g)$ can be expressed as a linear combination of elements in $\mathcal{I}$. We can directly check that 
		\begin{equation}
			\rho^2(\alpha\mbb)=\chi^2(\alpha\mbb)=\kappa^2(\alpha\mbb)=\alpha\mbb, \quad \forall \alpha\in \mathbb{C},
		\end{equation}
		and that for any element $\alpha e_{a_1}e_{a_2}\dots e_{a_k}$ with $k\in\mathbb{N}$, the same relation holds:
		\begin{align}
			&\rho^2(\alpha e_{a_1}e_{a_2}\dots e_{a_k})=\alpha\rho(e_{a_1}^\dagger e_{a_2}^\dagger \dots e_{a_k}^\dagger)=\alpha e_{a_1}e_{a_2}\dots e_{a_k},\\
			&\chi^2(\alpha e_{a_1}e_{a_2}\dots e_{a_k})=(-1)^{2k}\alpha e_{a_1}e_{a_2}\dots e_{a_k}=\alpha e_{a_1}e_{a_2}\dots e_{a_k},\\
			&\kappa^2(\alpha e_{a_1}e_{a_2}\dots e_{a_k})=(-1)^{2k}\alpha e_{a_1}e_{a_2}\dots e_{a_k}=\alpha e_{a_1}e_{a_2}\dots e_{a_k}.
		\end{align}
		This implies that $\rho^2=\chi^2=\kappa^2=\mbb$ on any element of $\mathcal{I}$, and thus on all of $Cl(V, \g)$.
	\end{proof}
	
	This implies the existence of positive and negative eigenspaces in each case
	\begin{align}
		&Cl(V, \g)=Cl^{\, +}(V, \g)\oplus Cl^{\, -}(V, \g), \quad \text{with} \quad \rho(a^\pm)=\pm a^\pm, \quad a^\pm\in Cl^{\, \pm}(V, \g),\\
		&Cl(V, \g)=Cl^{\, 0}(V, \g)\oplus Cl^{\, 1}(V, \g), \quad \text{with} \quad \chi(a^i)=(-1)^i a^i, \quad a^i\in Cl^{\, i}(V, \g),\\
		&Cl(V, \g)=Cl^{\, \bar{0}}(V, \g)\oplus Cl^{\, \bar{1}}(V, \g), \quad \text{with} \quad \kappa(a^i)=(-1)^i a^i, \quad a^i\in Cl^{\, \bar{i}}(V, \g).
	\end{align}

	\begin{proposition}
		\label{PropComStructAutom}
		We have $\rho\circ \chi=\chi \circ \rho$, $\rho\circ \kappa=\kappa \circ \rho$, and $\kappa\circ \chi=\chi \circ \kappa$ on $Cl(V, \g)$.
	\end{proposition}
	
	\begin{proof}
		These relations are easily verified on the first element $\alpha\mbb$ of $\mathcal{I}$ and on any element $v\in V$, since
		\begin{align}
			&\rho(\chi(v))=-v^\dagger=\chi(\rho(v)),\\
			&\rho(\kappa(v))=- v^\top=\kappa(\rho(v)),\\
			&\kappa(\chi(v))=\bar{v}=\chi(\kappa(v)).
		\end{align}
		Taking any element $a=\alpha e_{a_1}e_{a_2}\dots e_{a_k}$ with $k\in\mathbb{N}$, we obtain the same equalities, as the operators commute on each $e_{a_i} \in V$, ensuring that the relations hold on $\mathcal{I}$, then extend to all of $Cl(V, \g)$.
	\end{proof}
	
	The following proposition justifies the notation $\rho$ by establishing that it fulfills the defining properties of a twist i.e., satisfying the regularity condition.
	
	\begin{proposition}
		\label{PropRegKK}
		The automorphism $\rho$ is $\mathcal{B}(\mathcal{H})$-regular.
	\end{proposition}
	
	\begin{proof}
		Since $\rho^2=\mbb$ on all of $Cl(V, \g)$, we have $\rho^{-1}=\rho$. From the definition of $\rho$, taking any element $e_a\in V$, we obtain
		\begin{equation}
			\label{EqRegulDem}
			\rho(e_a^\dagger)=e_a=\rho(e_a)^\dagger,
		\end{equation}
		which implies that $\rho$ satisfies
		\begin{equation}
			\rho(e_a^\dagger) = \rho^{-1}(e_a)^\dagger.
		\end{equation}
		Thus, $\rho$ is regular on $V$. 
		
		We recover this relation for any element of $\mathcal{I}$, since $\rho(\mbb^\dagger)=\rho^{-1}(\mbb)^\dagger$, and for any $\alpha\in\mathbb{C}$ and $k\in\mathbb{N}$, we have
		\begin{align*}
			\rho((\alpha e_{a_1}e_{a_2}\dots e_{a_k} )^\dagger)&= \alpha^*\rho(e_{a_k}^\dagger)\dots \rho(e_{a_1}^\dagger)\\
			&= \alpha^*\rho^{-1}(e_{a_k})^\dagger\dots \rho^{-1}(e_{a_1})^\dagger\\
			&=\alpha^*\rho^{-1}(e_{a_1}\dots e_{a_k})^\dagger=\rho^{-1}(\alpha e_{a_1}e_{a_2}\dots e_{a_k})^\dagger.
		\end{align*}
		By linearity, this extends to all of $Cl(V, \g)\simeq M_{2^m}(\mathbb{C})\simeq \mathcal{B}(\mathcal{H})$, proving that $\rho$ is regular.
	\end{proof}
	
	We thus recover an automorphism with the properties of a twist, satisfying the regularity condition defined in section \ref{SecTwisted}, so the name twist for this operator.
	
	Since $\Aut(Cl(V, \g))=\Inn(Cl(V, \g))$, there exist operators $\tw, \Gamma, C\in Cl(V, \g)$ such that
	\begin{equation}
		\rho(\, \cdot\,)=\tw(\, \cdot\,)\tw^{-1},\qquad\quad  \chi(\, \cdot\,)=\Gamma(\, \cdot\,)\Gamma^{-1},\qquad\quad \kappa(\, \cdot\,)=C(\, \cdot\,)C^{-1}.
	\end{equation}
	We require these operators to be unitary, ensuring that the structural automorphisms are $*$-automorphisms.
	
	The real operator is defined from $C$ by 
	\begin{equation}
		J\defeq C\circ cc,
	\end{equation}
	so that the charge conjugation transformation is given by $\psi\to\psi_c\defeq J\psi$ with 
	\begin{equation}
		\label{RelJv}
		Jv=- vJ,\qquad\qquad \forall v\in V.
	\end{equation}
	From Proposition \ref{PropRegKK}, we retrieve the result of equation \eqref{EqFormTw}, namely $\tw=\exp(i\theta)\tw^\dagger$, since $\rho$ is $\mathcal{B}(\mathcal{H})$-regular. Additionally, we have:
	
	\begin{proposition}
		We have that $\overline{\tw}=\exp(i\theta^\prime)\tw$ with $\theta^\prime\in\mathbb{R}$.
	\end{proposition}
	
	\begin{proof}
		The relation $\rho(\bar{v})=\overline{\rho(v)}$ implies that 
		\begin{equation}
			\tw \bar{v}\tw^\dagger=\overline{\tw}\bar{v} \tw^\top.
		\end{equation}
		Equivalently, we write
		\begin{equation}
			\tw^\dagger\overline{\tw}\bar{v} \tw^\top\tw=\bar{v}, \qquad\quad \forall \bar{v}\in V.
		\end{equation}
		This implies that $\tw^\dagger\overline{\tw}$ belongs to the commutant $V^\prime$ and is therefore of the form $\alpha\mbb$ with $\alpha\in \mathbb{C}$. Multiplying on the left by $\tw$, we obtain
		\begin{equation}
			\overline{\tw}=\alpha\tw.
		\end{equation}
		Substituting into the previous expression, we find
		\begin{equation}
			\tw \bar{v}\tw^\dagger=\overline{\tw}\bar{v} \tw^\top=\alpha\alpha^*\tw\bar{v}\tw^\dagger,
		\end{equation}
		which requires $\alpha\alpha^*=1$, meaning that $\alpha=\exp(i\theta^\prime)$ with $\theta^\prime\in\mathbb{R}$.
	\end{proof}
	From now on, we impose $\overline{\tw}=\pm\tw$.
	
	\begin{proposition}
		We have that 
		\begin{equation}
			\tw\Gamma=\exp(i\theta_1)\Gamma\tw, \qquad \tw C=\exp(i\theta_2)C\tw, \qquad \text{and} \qquad \Gamma C=\exp(i\theta_3)C\Gamma,
		\end{equation}
		where $\theta_i\in \mathbb{R}$ for $i\in\{1,2,3\}$.
	\end{proposition}
	
	\begin{proof}
		From Proposition \ref{PropComStructAutom}, taking any element $a\in Cl(V, \g)$, we have
		\begin{equation}
			\label{Eqpomp}
			\tw\Gamma a (\tw\Gamma)^\dagger=\Gamma\tw a (\Gamma\tw)^\dagger, \quad 
			\tw C a (\tw C)^\dagger=C\tw a (C\tw)^\dagger, \quad 
			\Gamma C a (\Gamma  C)^\dagger=C\Gamma  a (C\Gamma )^\dagger.
		\end{equation}
		This implies that $\tw^{-1}\Gamma^{-1}\tw\Gamma$, $\tw^{-1}C^{-1}\tw C$, and $\Gamma^{-1}C^{-1}\Gamma C$ belong to the center $Z(Cl(V, \g))$, and are therefore of the form $\alpha\mbb$ with $\alpha\in\mathbb{C}$. As a consequence, we obtain
		\begin{equation}
			\tw\Gamma=\alpha_1\Gamma\tw, \qquad \tw C=\alpha_2C\tw, \qquad \text{and} \qquad \Gamma C=\alpha_3C\Gamma,
		\end{equation}
		where $\alpha_i\in\mathbb{C}$ for $i\in\{1,2,3\}$. Substituting these expressions back into equation \eqref{Eqpomp}, we find that $\alpha_i\alpha_i^*=1$, hence the result.
	\end{proof}
	To prevent unnecessary complex factors, we impose $\exp(i\theta_i)=\pm 1$ for $i\in\{1,2,3\}$.
	
	From now on, we assume the chosen representation of $e_a$ to be unitary\footnote{This choice is natural and essential for the upcoming definition of the Dirac operator.} with respect to the inner product $\langle \, \cdot\, , \, \cdot\, \rangle$. This, in particular implies that $e_a^\dagger=\g_a e_a$. Consequently, the vector spaces $V^\pm\subset Cl^{\, \pm}(V, \g)$ induced by $\rho$ are also those induced by the metric, as
	\begin{equation}
		\label{EqRelMetRho}
		\rho(e_a)=e_a^\dagger=\g_a e_a=\pm e_a\qquad\text{for}\qquad e_a\in V^\pm,
	\end{equation}
	hence the notation $Cl^{\, \pm}(V, \g)$ for the splitting induced by $\rho$. The twist then reverses the components associated with negative metrics. The automorphism $\rho$ thus corresponds to the parity operator, a point that will be further emphasized in Subsection \ref{subsectionsignaturechange}. 
	
	\begin{remark}
		\label{Rqcasriem}
		In the case of a Riemannian metric, i.e., $n=2m$, we have $e_a^\dagger=\g_a e_a=e_a$, so that $v=v^\dagger$, meaning that we can take $\tw=\mbb$.
	\end{remark}
	Let $\hat{V}$ be a vector space corresponding to a positive definite metric, such that the orthonormal basis $\{\hat{e}_a\}_{1\leq a\leq 2m}$ for $\hat{V}$ is connected to the one of $V$ by a Wick rotation:
	\begin{align*}
		&\hat{e}_a=e_a\qquad \qquad\,\, \text{for }  1\leq a\leq n,\\
		&\hat{e}_a=ie_a\qquad\qquad \text{for } n+1\leq a\leq 2m.
	\end{align*} 
	The unitarity of the $e_a$'s implies the unitarity of the $\hat{e}_a$'s. The corresponding real operator is denoted by $\JH=\CH\circ cc$\footnote{Similarly to subsection \ref{SecTwistedRealStruct}, as this corresponds to the spectral triple case.}, with $\hat{\kappa}(\hat{v})=-\bar{\hat{v}}$ for any $\hat{v}\in\hat{V}$.
	
	\begin{proposition}
		\label{PropRealPseudo}
		We have $\kappa=\hat{\kappa}\circ\rho=\rho\circ\hat{\kappa}$ on $Cl(V, \g)$.   
	\end{proposition}
	\begin{proof}
		Taking any element $e_a$ of the basis of $V$ we have
		\begin{align*}
			&\hat{\kappa}(\rho(e_a))=\hat{\kappa}(e_a)=\hat{\kappa}(\hat{e}_a)=- \bar{e}_a \qquad \qquad\qquad\quad\,\,\ \text{for }  1\leq a\leq n,\\
			&\hat{\kappa}(\rho(e_a))=-\hat{\kappa}(e_a)=i\hat{\kappa}(\hat{e}_a)=- \bar{e}_a\qquad\qquad\qquad\,\,\,\, \text{for } n+1\leq a\leq 2m.
		\end{align*} 
		Similarly, for $\rho\circ\hat{\kappa}$, using $\rho(\bar{v})=\overline{\rho(v)}$ we get
		\begin{align*}
			&\rho(\hat{\kappa}(e_a))=-\rho(\bar{e}_a)=-\overline{\rho(e_a)} =- \bar{e}_a \qquad\quad\quad  \qquad\,\,\, \text{for }  1\leq a\leq n,\\
			&\rho(\hat{\kappa}(e_a))=i\rho(\bar{\hat{e}}_a)=\rho(\bar{e}_a)=\overline{\rho(e_a)} =- \bar{e}_a\qquad\qquad  \,\, \text{for } n+1\leq a\leq 2m.
		\end{align*} 
		Thus, we obtain $\kappa=\hat{\kappa}\circ\rho=\rho\circ\hat{\kappa}$ on all $V$. Then, for any element $a=\alpha e_{a_1}e_{a_2}\dots e_{a_k}$ with $\alpha\in\mathbb{C}$ and $k\in\mathbb{N}$, we have:
		\[
		\kappa(a)=\rho(\hat{\kappa}(a))=\hat{\kappa}(\rho(a)).
		\]
		By linearity, this relation extends to all of $Cl(V, \g)$, concluding the proof.
	\end{proof}

	We can then define $C=\tw\hat{C}$ with $\tw\hat{C}=\pm \hat{C} \tw$ as $\tw C=\pm C \tw$. We recover the relation
	\begin{equation}
		\label{EqDelta}
		J=\tw\JH 
	\end{equation}
	obtained in the context of twisted spectral triples (see equation \eqref{EqRelJ}).
	
	\begin{remark}
		In full generality, if we do not restrict to the even-dimensional case, the defining relation for $\kappa$ is given by $\kappa(v)\defeq -\epsilon_1^\tw\bar{v}$ (see relation \eqref{EqDefStructParamK} for the definition of $\epsilon_1^\tw$). By defining $\hat{\kappa}(\hat{v})=-\hat{\epsilon}_1\bar{\hat{v}}$ and $\delta=\epsilon_1^\tw/\hat{\epsilon}_1=\pm 1$, the result of Proposition \ref{PropRealPseudo} becomes
		\[
		\kappa=\delta\hat{\kappa}\circ\rho=\delta\rho\circ\hat{\kappa}.
		\]
		The result of equation \eqref{EqDelta} is then a specificity of the case $\delta=1$, which always holds in the even-dimensional case (see Table \ref{KoDimTable}).
	\end{remark}
	
	In addition, we have the relation
	\begin{equation}
		\tw J=\tw C\circ cc=\pm C \tw\circ cc=\pm J\overline{\tw}=(\pm) J\tw,
	\end{equation}
	where the last $\pm$ sign is in parentheses to signify that it is not necessarily the same $\pm$ sign as the ones without parentheses. This notation will not be retained in what follows. We recover relation \eqref{EqRelKJ}, i.e., $\rho(J)=\epsilon J$, given in the context of twisted spectral triples.
	
	It is interesting to note that most of these results would have been different in the odd-dimensional case, as the Clifford algebra is not simple in this case. This implies, in particular, that not all automorphisms are inner, and the commutant $V^\prime$ is no longer proportional to $\mathbb{C}$, properties that have been essential in many demonstrations.

	\subsection{From automorphisms of Clifford algebras to Krein products}
	\label{subsecAutomClifftoKrein}
	
	This section aims to show how the previously defined twist plays a fundamental role in describing the automorphisms of the Clifford algebra $Cl(V, \g)$ which are also automorphisms of $V$ and the induced automorphism invariant inner product in the case of the orthochronous spin group.

	The following proposition establishes the $\glv$-invariance of the structural automorphisms' definitions, justifying the term "structural". 
	
	\begin{proposition}
		\label{PropConDefStruct}
		The relations given in \eqref{EqDefStruct} are preserved under $\glv$, i.e.,
		\begin{equation}
			\label{EqConsRelStructOp}
			\rho(\phi(v))=\phi(v)^\dagger, \qquad\quad \chi(\phi(v))=-\phi(v), \qquad\quad \kappa(\phi(v))=-\overline{\phi(v)}
		\end{equation}
		for all $v\in V$ and all $\phi\in \glv$.
	\end{proposition}
	\begin{proof}
		Taking $\phi\in \glv$, from relation \eqref{RelExtent} $\phi$ acts as a real linear transformation on $V$, so that $\phi(e_a)=\sum_{b=1}^{2m}\lambda_{ab}e_b$ for some $\lambda_{ab}\in \mathbb{R}$. Then we have that
		\begin{align*}
			&\rho(\phi(e_a))=\sum_{b=1}^{2m}\lambda_{ab}\rho(e_b)=\sum_{b=1}^{2m}\lambda_{ab}e_b^\dagger=(\sum_{b=1}^{2m}\lambda_{ab}e_b)^\dagger=\phi(e_a)^\dagger, \\
			&\chi(\phi(e_a))=\sum_{b=1}^{2m}\lambda_{ab}\chi(e_b)=-\sum_{b=1}^{2m}\lambda_{ab}e_b=-\phi(e_a), \\
			&\kappa(\phi(e_a))= \sum_{b=1}^{2m}\lambda_{ab}\kappa(e_b)=-\sum_{b=1}^{2m}\lambda_{ab}\bar{e}_b=-\overline{\phi(e_a)}.
		\end{align*}
		We then recover the result of relations \eqref{EqConsRelStructOp} for any $v\in V$.
	\end{proof}
	Notably, if $\phi\in \glv$, we have $\phi\circ \chi=\chi\circ \phi$.
	
	From now on, we will only consider elements in $\glpv=\opvg\subset\Aut(Cl(V, \g))$.
	
	\begin{proposition}
		Given $\phi \in \opvg$, defined via $x\in Cl(V, \g)^{\, \times}$ by $\phi_x(y)=\ad_x(y)$ for all $y\in Cl(V, \g)$, we have
		\begin{equation}
			\label{EqGh}
			x^{-1}=\alpha \rho(x^\dagger),
		\end{equation}
		where $\alpha$ is a real number.
	\end{proposition}
	
	\begin{proof}
		Since $\phi\in \glv$, the preservation of the defining relation $\rho(v)=v^\dagger$ implies that $\phi_x(v)^\dagger=\rho(\phi_x(v))$ for any $v\in V$, so that
		\begin{equation}
			\label{EqPA}
			(x^{-1})^\dagger v^\dagger x^\dagger=\tw x v x^{-1} \tw^\dagger.
		\end{equation}
		Using the fact that $v^\dagger=\tw v\tw^\dagger$, equation \eqref{EqPA} becomes
		\begin{equation*}
			x^{-1}\tw^\dagger (x^{-1})^\dagger \tw v\tw^\dagger x^\dagger\tw x=v.
		\end{equation*}
		This rewrites $A^{-1} v A = v$, where $A=\tw^\dagger x^\dagger\tw x$. This implies that $A$ belongs to the commutant of $V$, and hence takes the form $\beta \mbb$ from Proposition \ref{PropComV}, with $\beta\in\mathbb{C}$. Consequently, we obtain $\tw^\dagger x^\dagger\tw x=\beta\mbb$ and then $	x^{-1}=\frac{1}{\beta} \tw^\dagger x^\dagger \tw$.	Multiplying $\tw^\dagger x^\dagger\tw x=\beta\mbb$ on the left by $\tw$, this relation becomes $x^\dagger\tw x=\beta\tw$. Taking the adjoint, we obtain $x^\dagger\tw^\dagger x=\beta^*\tw^\dagger$, which is equivalent to $x^\dagger\tw x=\beta^*\tw$
		since $\tw=\exp(i\theta)\tw^\dagger$. This implies that $\beta$ is real. Identifying $\alpha$ with $1/\beta$, we obtain $x^{-1}=\alpha \tw^\dagger x^\dagger \tw=\alpha \tw x^\dagger \tw^\dagger$, hence the result.
	\end{proof}
	We define $V^{\, \times}=V\cap Cl^{\, \times}(V, \g)$ to be the subset of invertible elements in $V$ or equivalently\footnote{The relation $vv=\g_v\mbb$ directly implies the existence of $v^{-1}$ if $\g_v\neq 0$. If $\g_v= 0$ we have that $vv=0$, so that the existence of $v^{-1}$ would lead to the contradiction $v=v(vv^{-1})=(vv)v^{-1}=0$, meaning that $v^{-1}$ does not exist in this case. Hence, the equivalent definitions for $V^{\, \times}$.} 
	\begin{align}
		V^{\, \times}=\{v\in V\, \mid\, \g(v,v)\neq 0\}.
	\end{align}
	
	Taking $u,v\in V$ and observing that $\ad_v(u)\in V$, the group $P(V, \g)\subset \opvg$ is defined as the group generated by the elements of $V^{\, \times}$:
	\begin{equation}
		P(V, \g)=\{v_1\dots v_k\, \mid\, v_1, \dots , v_k \in V^{\, \times} \quad \text{and}\quad k\in \mathbb{N}^*\,\, \text{is finite}\}.
	\end{equation}
	
	Taking $x=v_1\dots v_k\in P(V, \g)$, we define $q(x)\defeq \card\{v_i\, \mid\, \g_{v_i}<0\}$. We then have that 
	\begin{equation*}
		x^{-1}=v_k^{-1}\dots v_1^{-1}=v_k g_{v_k}^{-1}\dots v_1 g_{v_1}^{-1}
	\end{equation*}
	
	and that $\rho(x^\dagger)=v_k\dots v_1$, implying that equation \eqref{EqGh} becomes
	\begin{equation}
		\label{EqRelInt}
		x^{-1}=\alpha \rho(x^\dagger)=g_{v_1}^{-1}\dots g_{v_k}^{-1}\rho(x^\dagger).
	\end{equation}
	
	Now if we define the (twisted) adjoint action $\tphi_x(y)\defeq\tad_x(y)=\chi(x) y x^{-1}$, taking $u, v \in V$, we have that
	\begin{equation}
		\tad_v(u)=u-2\frac{\g(v, u)}{\g_v}v 
	\end{equation}
	is a reflection across the orthogonal complement $v^\perp$. Moreover, we have 
	\[
	\tad_{x}=\tad_{v_1\dots v_k}=\tad_{v_1}\circ \dots \circ \tad_{v_k} \quad \text{and} \quad \ad_{v_1\dots v_k}=\ad_{v_1}\circ \dots \circ \ad_{v_k}
	\]
	since $\tad_v=-\ad_v$. Therefore, the image of $P(V,\g)$ under $\tad$ is the group generated by reflections, which is the full orthogonal group; see \cite{lawson2016spin} for more details.

	The unit sphere $\US$ for $\g$ is defined as 
	\begin{align*}
		\US\defeq\{v\in V\, \mid\, \g_v=\pm 1\}.
	\end{align*}
	Taking ${v}\in V\cap \US$ and $u\in V$, we obtain
	\begin{equation}
		\tad_{v}(u)=u-2\frac{\g({v}, u)}{\g_{v}}{v}=\left\{
		\begin{array}{ll}
			\, \,\,\,	u &\qquad\text{ for $u\in v^\perp$}  ,\\ 
			-u &\qquad\text{ for $u=v$.}
		\end{array}\right.
	\end{equation} 
	Such a transformation corresponds precisely to an orthogonal reflection. The induced subgroup $Pin(V, \g)\subset P(V, \g)$ is defined as 
	\begin{equation}
		\label{EqPin}
		Pin(V, \g)=\{v_1\dots v_k\, \mid\, v_1, \dots , v_k \in \US \quad \text{and}\quad k\in \mathbb{N}^*\}.
	\end{equation}
	Taking any $x\in Pin(V, \g)$, from equation \eqref{EqRelInt} we have that 
	\begin{align*}
		x^{-1}=\pm \rho(x^\dagger).
	\end{align*}

	We define the "fixed" $\g$-orthonormal basis 
	\[
	e_{(i)}\defeq e_a|_{a=i} 
	\]
	which is, by definition, invariant under $\glv$. The corresponding "fixed" vector space is denoted by $V_f$. The subgroup $Pin(V_f, \g)$, i.e., the fixed version of $Pin(V, \g)$, is generated by the elements of $V_f$ in the unit sphere.

	\begin{proposition}
		The structural automorphisms $\rho, \chi, \kappa$ are elements of $Pin(V_f, \g)$ represented by the adjoint action.
	\end{proposition} 
	\begin{proof}
		The automorphisms $\rho, \chi, \kappa$ are elements of $\opvg$. Taking the general notation $\phi$ for any of these automorphisms, from Proposition \ref{PropInvol}, we have that $\phi$ is an involution on $Cl(V, \g)$. As a consequence, we can define the eigenspaces $W^\pm\subset V$ so that 
		\begin{align*}
			V=W^+\oplus W^-, \qquad\qquad  \text{with } \phi(w^\pm)=\pm w^\pm \text{ for } w^\pm\in W^\pm.
		\end{align*}
		This implies that $\phi$ acts as an orthogonal reflection across $W^+$. From Proposition \ref{PropConDefStruct}, the definition of $\phi$ is $\glv$-invariant. 
		
		By defining $\{f_{(i)}\}_{1\leq i\leq l}$ and $\{f_{(i)}\}_{l+1\leq i\leq 2m}$ as fixed $\g$-orthonormal bases of $W^+$ and $W^-$ respectively, then for all $u\in V$, we have
		\begin{equation*}
			\phi(u)=\left\{
			\begin{array}{ll}
				\ad_{f_{(1)}}\circ \dots \circ \ad_{f_{(l)}}(u)=\ad_{f_{(1)}\dots f_{(l)}}(u) &\quad\text{for } l \text{ odd},  \\ 
				(\mbb)\ad_{f_{(l+1)}}\circ \dots \circ \ad_{f_{(2m)}}(u)=\ad_{f_{(l+1)}\dots f_{(2m)}}(u) &\quad\text{for } l \text{ even}.
			\end{array}\right. 
		\end{equation*}
		where $f_{(1)}\dots f_{(l)}$ and $f_{(l+1)}\dots f_{(2m)}$ are elements of $Pin(V_f, \g)$, thus providing the desired $\glv$-invariant automorphism $\phi$, hence the result.
	\end{proof}
	For $\rho$, noting that $W^\pm=V^\pm$, the unitary operator $\tw$ is then given by 
	\begin{equation}
		\label{EqNewK}
		\tw=\left\{
		\begin{array}{ll}
			\alpha_1\prod_{i=1}^{n} e_{(i)} &\qquad\text{for } n \text{ odd},  \\ 
			\alpha_1\,  (\mbb)\prod_{i=n+1}^{2m} e_{(i)} &\qquad\text{for } n \text{ even}.
		\end{array}\right.  
	\end{equation} 
	with $\alpha_1^{-1}=\alpha_1^*$ a complex number, where the fixed basis $e_{(i)}$ guarantees that $\tw$ remains $\glv$-invariant.

	For $\Gamma$, we have that $V=W^-$, since $\chi(v)=-v$ for any $v\in V$, so that
	\begin{equation}
		\label{EqNewG}
		\Gamma= \alpha_2\prod_{i=1}^{2m} e_{(i)}
	\end{equation}
	with $\alpha_2^{-1}=\alpha_2^*$ a complex number. We now fix $\alpha_2=i^{-m(2m-1)-n}$ so that $\Gamma^2=\mbb$ and $\Gamma^\dagger=\Gamma$.

	For $C$, a convenient choice is to require that the basis $\{e_a\}_{1\leq a\leq 2m}$ (which is then the fixed one) satisfies $\bar{e}_a=\pm e_a$. By defining the index set $I\defeq \{i\, \mid\, \kappa(e_{(i)})=-\bar{e}_{(i)}=e_{(i)} \}$
	corresponding to the elements in $W^+$, and $\bar{I}\defeq \{1, \ldots, 2m\} \setminus I$ for $W^-$, we obtain
	\begin{equation}
		\label{Cexplicite}
		C=\left\{
		\begin{array}{ll}
			\alpha_3\prod_{i\in I} e_{(i)} &\qquad\text{for } \card(I) \text{ odd},  \\ 
			\alpha_3\,  (\mbb)\prod_{i \in  \bar{I}} e_{(i)} &\qquad\text{for } \card(I) \text{ even}.
		\end{array}\right.  
	\end{equation} 
	with $\alpha_3^{-1}=\alpha_3^*$ a complex number.
	
	The invariance of these three operators under $\glv$ follows from their definition based on the fixed basis, in accordance with Proposition \ref{PropConDefStruct}.

	Coming back to the group $Pin(V, \g)$, the associated spin group is defined as
	\begin{align*}
		Spin(V, \g)=Pin(V, \g)\cap Cl^{\, 0}(V, \g),
	\end{align*}
	or equivalently, as the subset of $Pin(V, \g)$ corresponding to even numbers $k$. The representations of $Spin(V, \g)$ are called spinor representations and will be denoted by $\Sp$. For any $x\in Spin(V, \g)$, the image under the adjoint and twisted adjoint actions coincide, i.e., $\tad_x=\ad_x$.
	
	The subgroup $Spin(V, \g)^+\subset Spin(V, \g)$ called orthochronous spin group is defined as the set of elements of $Spin(V, \g)$ that preserve the orientation of $V^\pm$ independently (spatial and temporal). Taking ${v^\pm}\in V^\pm\cap \US$ and $u\in V$, we have that 
	\begin{equation}
		\tad_{v^\pm}(u)=u-2\frac{\g({v^\pm}, u)}{\g_{v^\pm}}{v^\pm}=\left\{
		\begin{array}{ll}
			\, \,\,\,	u &\qquad\text{for } u\in (v^\pm)^\perp,  \\ 
			-u &\qquad\text{for } u=v^\pm.
		\end{array}\right. 
	\end{equation} 
	meaning that $\tad_{v^\pm}$ reverses the orientation of $V^\pm$ while leaving the one of $V^\mp$ unchanged. Taking $x=v_1\dots v_k\in Spin(V, \g)$, the transformation $\tad_{x}$ preserves the orientation of $V^\pm$ only if $q(x)$ is even, since 
	\[
	\tad_{v_1\dots v_k}=\tad_{v_1}\circ \dots \circ \tad_{v_k}.
	\]
	As a consequence, the orthochronous spin group corresponds to the subgroup of elements $x\in Spin(V, \g)$ such that $q(x)$ is even:
	\begin{equation}
		Spin(V, \g)^+=\{x=v_1\dots v_k\, \mid\, v_1, \dots , v_k \in \US, \quad q(x) \in 2\mathbb{N}, \quad \text{and}\quad k\in 2\mathbb{N}^*\}.
	\end{equation}
	For any $x\in Spin(V, \g)^+$, we directly obtain
	\begin{equation}
		\label{EqformSpin}
		x^{-1}= \rho(x^\dagger)
	\end{equation}
	from equation \eqref{EqRelInt}, since $q(x)$ is even.
	
	Consequently, a natural $Spin(V, \g)^+$-invariant inner product is given by the $\tw$-product $\inntw$ defined in equation \eqref{Eqtwprod} from $\langle \, \cdot\, , \, \cdot\, \rangle$. The associated adjoint operator $\Tadj$ is unique, as $\rho$ is $\calB(\calH)$-regular. equation \eqref{EqformSpin} then rewrites as $x^{-1}= x^\Tadj$, implying that any $x\in Spin(V, \g)^+$ is a $\tw$-unitary operator. The relation $\rho(v)=v^\dagger$ then becomes
	\begin{equation}
		\label{RelVadj}
		v=v^\Tadj, \qquad\qquad \forall v\in V,
	\end{equation}
	which is preserved under $\glv$ and consequently under $Spin(V, \g)^+$.

	The inner product $\inntw$ allows the definition of the dual/adjoint spinor of $\psi\in \Sp$ as
	\begin{equation}
		\overline{\psi}=\psi^\dagger \tw
	\end{equation}
	through the relation $\overline{\psi}(\phi)=\langle \psi ,\phi \rangle_\tw$.

	\begin{remark}
		The specific notation $\inntw$ emphasizes the fact that the operator $\tw$, being $\glv$-invariant, is embedded within the algebraic structure. In the case $n=2m$, i.e., a positive definite metric, we can take $\tw=\mbb$, as mentioned in Remark \ref{Rqcasriem}. In this case, we recover the usual Hermitian inner product $\inntw=\langle \, \cdot\, , \, \cdot\, \rangle$.
	\end{remark}

%	\begin{definition}
%		Given a vector $\psi\in \Sp$ and the $\tw$-product $\inntw$, we define the function $\|\, \cdot\,  \|_\tw$ by\GN{supp}
%		\begin{equation}
%			\label{DefNormGen}
%			\|\psi \|_\tw\defeq \sqrt{\langle\psi ,\psi \rangle_\tw}.
%		\end{equation}
%	\end{definition}
%	This generalizes the usual notion of norm, which is retrieved when $\tw=\mbb$.

	As stated in Subsection \ref{subsecFromTwToIndef}, an indefinite (or definite) Hermitian product $\inntw$ can only be obtained by eliminating the phase shift in equation \eqref{EqFormTw}, enforcing $\tw=\tw^\dagger$ so that $\inntw$ is a Krein product. In what follows, we will only consider such fundamental twists, ensuring that any product $\inntw$ is indefinite and Hermitian, with
	\begin{equation}
		\tw=\tw^\dagger=(\overline{\tw})^\top=\pm\tw^\top=\pm \overline{\tw}.
	\end{equation}

	\subsection{Even dimensional pseudo-Riemannian manifolds} 
	\label{SubsecPseudoRiem}
	
	Let's consider a $2m$-dimensional smooth manifold $\Man$ equipped with a non-degenerate metric $\g$ of signature $(n, 2m-n)$, defined on the tangent space $\tm$.

	At each point $x\in\Man$, we associate the local metric $\g_x$ and the tangent space $\txm$. The Clifford bundle $Cl_{n,\, 2m-n}$ is the bundle of Clifford algebras, defined pointwise as
	\[
	Cl_{n,\, 2m-n}=\bigcup_{x\in\Man}Cl(\txm, \g_x),
	\]
	where $Cl(\txm, \g_x)$ is the Clifford algebra generated by the elements $\cl(v)$ and $\mbb$ for $v\in\txm$, satisfying the relation 
	\begin{equation}
		\cl(u)\cl(v) + \cl(v)\cl(u) = 2\g_x(u, v)\mbb,\qquad \qquad \forall u, v \in \txm.
	\end{equation}
	
	We define the vector space $V_x\subset Cl(\txm, \g_x)$ associated with $\txm$ by
	\begin{equation}
		V_x\defeq \{\cl(v)\, \mid\, v\in \txm\}
	\end{equation}
	so that $Cl(\txm, \g_x)$ is generated by $V_x$ and $\mbb$. The corresponding vector bundle is then given by $V=\bigcup_{x\in\Man} V_x$. As previously, unless explicitly mentioned otherwise, we replace the notation $\cl(v)$ by $v$ for any element $v\in V_x$ and for any $x\in\Man$. Automorphisms of the vector bundle $V$ are defined as  
	\begin{equation}
		\glv=\bigcup_{x\in\Man} GL(V_x).
	\end{equation}
	The spinor bundle $\Sp$ over $\Man$ is a left-module for $Cl_{n,\, 2m-n}$. A section $\psi :\Man\to \Sp$ of $\Sp$ is a map that assigns to each point $x\in \Man$ a fiber $\Sp_x$. These sections, known as spinor fields, constitute the vector space denoted $\Gamma(\Sp)$. In this framework, the Spin group is more commonly denoted as $Spin(n, 2m-n)$.

	We define the Hermitian inner product 
	\[
	\inn : \Gamma(\Sp) \times \Gamma(\Sp)\to C^\infty(\Man, \mathbb{C})
	\]
	with adjoint $\dagger$. As $\Man$ is locally compact, any vector bundle carries a continuously varying fiber-wise inner product $\inn_x$ so that for $\psi, \psi^\prime\in\Gamma(\Sp)$, we have $\langle \psi, \psi^\prime\rangle_x=\langle \psi(x), \psi^\prime(x)\rangle_x$.

	The twist is (locally) defined on $V_x$ by the relation
	\begin{equation}
		\rho_x(v)=\twx v\twx^\dagger\defeq\rho(v)=v^\dagger, \qquad\qquad  \forall v\in V_x,
	\end{equation}
	with $\rho_x(w)=0$ for any $w\in V_{y\neq x}$ and $\twx\defeq \tw|_x$.
	
	From this, we define the fiber-wise indefinite inner product
	\begin{equation}
		\langle \, \cdot\, , \, \cdot\, \rangle_{\twx}=\langle \, \cdot\, , \twx\, \cdot\, \rangle_x,
	\end{equation}
	which is invariant under $Spin(V_x, \g_x)^+$. 
	
	Using the fiber-wise evaluation $\|\psi  \|_{x}^2=\langle \psi(x) , \psi(x) \rangle_{x}$, we define its bundle counterpart as
	\begin{equation}
		\|\psi  \|^2\defeq \int_\Man\|\psi  \|_{x}^2\sqrt{|\g|}dx,
	\end{equation}
	where $\sqrt{|\g|}$ is the volume form.

	The space of square-integrable sections of the spinor bundle over $\Man$ is defined by
	\begin{align}
		\label{Ldeux}
		L^2(\Man, \Sp)\defeq \{\psi\in \Gamma(\Sp)\, \mid\,\, \|\psi  \|^2<\infty \}.
	\end{align}
 A local oriented $\g$-orthonormal basis of $\tm$ is given by $\{\partial_1, \dots, \partial_{2m}\}$, with components satisfying $g_{ab}=\g(\partial_a, \partial_b)$. The basis of the dual tangent space $T^*\Man$ is given by $\{dx^1, \dots, dx^{2m}\}$, defined by the relation $\g(dx^b, \partial_a)=\delta^b_a$. The metric $\g$ induces a splitting $\tm =\tm^+\oplus \tm^-$, where a local oriented orthonormal basis is given by $\{\partial_1, \dots, \partial_{n}\}$ for $\tm^+$ and $\{\partial_{n+1}, \dots, \partial_{2m}\}$ for $\tm^-$. 
 
	Defining $\g^{ab}=\g(dx^a,dx^b)$ in the $\g$-orthonormal basis, the gamma matrices are defined as the unitary operators $\gamma^a\defeq \cl(dx^a)$ for all $a\in\{1,\dots, 2m\}$, satisfying
	\begin{equation}
		\label{EqMetGamma}
		\{\gamma^a, \gamma^b\}=2\g^{ab}\mbb.
	\end{equation}
	The relation to general coordinates $dx^\mu$\footnote{denoted by Greek indices} is provided by the vielbein $e^\mu_a\in\CM$:
	\begin{equation}
		\label{EqRelVielb}
		\cl(dx^\mu)=\gamma^\mu=\cl(e^\mu_a dx^a)=e^\mu_a\gamma^a.
	\end{equation} 
	The relation between the twist and the adjoint remains valid in general coordinates
	\begin{equation}
		(\gamma^\mu)^\dagger=e^\mu_a(\gamma^a)^\dagger=e^\mu_a\rho(\gamma^a) = \rho(\gamma^\mu).
	\end{equation}
	This implies that the operator $\tw$ is independent of the chosen coordinate system. Consequently, the relation in equation \eqref{RelVadj} became
	\begin{equation}
		(\gamma^\mu)^\Tadj=\gamma^\mu,
	\end{equation}
	which is preserved under the action of $\glv$.
	
	The Dirac operator is then given by the $\Tadj$-selfadjoint operator 
	\begin{equation}
		\label{EqDefDirPseudo}
		\KDir=i\gamma^\mu \nabla_\mu^{{\scriptscriptstyle\Sp}},
	\end{equation} 
	where the spin connection $\nabla_a^{{\scriptscriptstyle\Sp}}=\partial_a+\frac{1}{4}\Gamma^b_{\,\, a c}\gamma^c\gamma_b$ is the lift of the Levi-Civita connection to the spinor bundle. 
	
	Using the $\glv$-invariant matrices $\gamma^{(i)}\defeq \gamma^a|_{a=i}$, from equation \eqref{EqNewK} the operator $\tw=\tw^\dagger$ takes the explicit form
	\begin{equation}
		\label{ExplicitK}
		\tw=\left\{
		\begin{array}{ll}
			i^{-n(n-1)/2}\prod_{i=1}^{n} \gamma^{(i)} &\qquad\text{for } n \text{ odd},  \\ 
			(\mbb)\prod_{i=n+1}^{2m} \gamma^{(i)} &\qquad\text{for } n \text{ even}.
		\end{array}\right.  
	\end{equation}
	Similarly, a grading operator satisfying $\Gamma^\dagger=\Gamma$ is defined from equation \eqref{EqNewG} as
	\begin{equation}
		\label{DefOpGrad}
		\Gamma=i^{-m(2m-1)-n} \prod_{i=1}^{2m} \gamma^{(i)}.
	\end{equation}
	The real structure operator $J$ is defined from $\tw$ and $\JH$ following equation \eqref{EqRelJ}. From equation \eqref{RelJv}, we obtain relation
	\begin{equation}
		J\gamma^\mu =-\gamma^\mu J.
	\end{equation}
	
	Consequently, the operators $\tw, \Gamma$, and $J$ remain invariant under $\glv$, encompassing coordinate transformations and the Spin group action.
	
	\begin{proposition}
		\label{PropRelKJG}
		In the case where $n$ is odd, we have 
		\begin{equation}
			\label{EqRelKG}
			\tw\Gamma=-\Gamma\tw, \qquad\text{and}\qquad \tw J=(-1)^{n(3n-1)/2}J\tw.
		\end{equation}
		In the case where $n$ is even, we obtain
		\begin{equation}
			\tw\Gamma=\Gamma\tw, \qquad\text{and}\qquad \tw J=J\tw.
		\end{equation}
	\end{proposition}

	\begin{proof}
		In the even case, we have
		\begin{align}
			\tw\Gamma=(-1)^n\Gamma \tw=-\Gamma \tw.
		\end{align}  
		Noting that $i^{-n(n-1)/2}=i^{3n(n-1)/2}$ and that $J i=-iJ$, we obtain
		\begin{align}
			J\tw=(-1)^{3n(n-1)/2}i^{3n(n-1)/2}J\prod_{i=1}^{n} \gamma^{(i)}=(-1)^{3n(n-1)/2}(-1)^{n}\tw J=(-1)^{n(3n-1)/2}\tw J.
		\end{align}
		Similarly, in the odd case, we obtain $\tw\Gamma=(-1)^{2m-n}\Gamma \tw=\Gamma \tw$ and $J\tw=(-1)^{2m-n} \tw J=\tw J$.
	\end{proof}
	
	From equation \eqref{EqRelKJ}, we specify the relation between $J$ and $\tw$ using the parameter $\epsilon$:
	\begin{equation}
		\tw J=\epsilon J\tw,
	\end{equation}
	and the relation between $\tw$ and $\Gamma$ using the parameter $\epsilon^\prime=\pm 1$:
	\begin{equation}
		\label{EqRelKGamma}
		\tw \Gamma=\epsilon^\prime \Gamma\tw.
	\end{equation}
	Note that $\epsilon$ and $\epsilon^\prime$ are functions of $n$ in all cases.
	
	\begin{remark}
		\label{RQZZ}
		The odd case, where $\epsilon^\prime=-1$, corresponds to a twist by the grading, since it implies $\rho(\Gamma)=-\Gamma$ so that $\rho$ exchanges the representations $\pi_+$ and $\pi_-$ in equation \eqref{EqRhogen}.
	\end{remark}
 	
 We then get the pseudo-Riemannian spectral triple
 \begin{align*}
 (\CM,\, (L^2(\Man, \Sp),\, \inntw),\, \KDir,\, J,\, \Gamma).
 \end{align*}

	Now, let us focus on the four-dimensional Lorentzian case. The four gamma matrices (in the chiral representation) are given by $\gamma^a$ with $a\in \{0, 1, 2, 3\}$,\footnote{The index $a$ now starts at $0$ instead of $1$, aligning with the standard physics notation.} defined as follows:
	\begin{equation}
		\label{EDirac}
		\gamma^a =
		\left( \begin{array}{cc}
			0 & \sigma^a \\ \tilde\sigma^a & 0
		\end{array} \right).
	\end{equation}
	
	where, for $a= 0$ and $j\in \{1,2,3\}$, we define $\sigma^a \coloneqq  \left\{ \mathbb{\bbbone}_2, -\sigma_j \right\}$ and $\tilde\sigma^a \coloneqq  \left\{ \mathbb{\bbbone}_2, \sigma_j \right\}$, with the Pauli matrices
	\begin{equation}
		\label{Pauli}
		\sigma_1 = \left(\begin{array}{cc} 0 & 1 \\ 1 & 0 \end{array}\right)\!,
		\qquad\quad	\sigma_2 = \left(\begin{array}{cc} 0 & -i \\ i & 0 \end{array}\right)\!,
		\qquad\quad		\sigma_3 = \left(\begin{array}{cc} 1 & 0 \\ 0 & -1 \end{array}\right)\!.
	\end{equation}
	They satisfy the anticommutation relation
	\begin{equation}
		\gamma^a\gamma^b + \gamma^b \gamma^a =2\eta^{ab}\mathbb \bbbone_4, \qquad\quad	
		\forall a, b= 0, ..., 3,
	\end{equation}
	where $\eta^{ab}$ is the Minkowski metric. The grading operator is given by
	\[
	\Gamma=-i\gamma^{(0)}\gamma^{(1)}\gamma^{(2)}\gamma^{(3)}.
	\]
	Noting that $\gamma^2$ is the only gamma matrix such that $\bar{\gamma}^2=-\gamma^2$, from equation \eqref{Cexplicite}, the real structure is defined as
	\[
	J=-i\gamma^{(2)}\circ cc.
	\]
	Taking $\tw=\gamma^{(0)}$, the indefinite inner product is given by $\inntw$, with the adjoint  
	\begin{equation}
		(\, \cdot\,)^\Tadj\defeq \gamma^{(0)}(\, \cdot\,)^\dagger\gamma^{(0)},
	\end{equation}
	so that the Lorentzian Dirac operator $\KDir$ is $\tw$-selfadjoint. 
	
	If $S\in Spin(1,3)^+$ is an element of the orthochronous Lorentz group, then $S$ is unitary with respect to this inner product, i.e.,
	\begin{equation}
		SS^\Tadj=S^\Tadj S=\bbbone_4.
	\end{equation}
	These transformations naturally preserve the Dirac Lagrangian 
	\begin{equation}
		\label{EqLagDir}
		\lag_f=\overline{\psi} \KDir \psi+ m\overline{\psi} \psi.
	\end{equation}
	
	Lorentz transformations act on coordinates and partial derivatives as 
	\begin{equation}
		(x^a)^\prime=\Lambda^a_{\,\, b}x^b, \quad \text{and} \quad \partial_a\to \partial_a^\prime=(\Lambda^{-1})^b_{\,\, a}\partial_b.
	\end{equation}
	Taking $S(\Lambda)$ as the bispinor representation of an arbitrary Lorentz transformation $\Lambda$, we have that $\psi^\prime(x^\prime)=S(\Lambda)\psi(x)$, with the relation
	\begin{equation}
		S\gamma^a S^{-1}=(\Lambda^{-1})^a_{\,\, b}\gamma^b.
	\end{equation}
	For the Dirac operator $\KDir=i\gamma^a\partial_a$, Lorentz transformations act as
	\[
	\KDir\,\to\, S\KDir S^{-1}.
	\]
	This transformation can be viewed in two equivalent ways:
	
	\begin{enumerate}
		\item As a transformation of the partial derivatives:
		\begin{equation}
			\partial_a\to \partial_a^\prime=(\Lambda^{-1})^b_{\,\, a}\partial_b, \qquad\quad	 \gamma^a\to\gamma^a, \qquad\quad	 \gamma^{(0)}\to \gamma^{(0)}.
		\end{equation} 
		
		\item As a transformation of the Clifford algebra:
		\begin{equation}
			\partial_a\to \partial_a^\prime=\partial_a, \qquad\quad	 \gamma^a\to\gamma^{\prime a}=S\gamma^a S^{-1}, \qquad\quad	 \gamma^{(0)}\to \gamma^{(0)}.
		\end{equation} 
	\end{enumerate}
	A distinction between the transformation behaviors of the matrices $\gamma^{0}$ and $\gamma^{(0)}$ arises only in the second case, this justifies the need to distinguish these two matrices.
	
	The notation $\overline{\psi}=\psi^\dagger \gamma^0$, commonly used in quantum field theory textbooks instead of $\overline{\psi}=\psi^\dagger \gamma^{(0)}$, is therefore misleading. It suggests that $\tw=\gamma^0$ is a true gamma matrix, subject to coordinate changes and varying under Lorentz transformations when applied to the Clifford algebra. However, this does not pose an issue in standard quantum field theory textbooks, as they typically restrict themselves to Minkowski spacetime, where the Clifford algebra is fixed in the background and Lorentz transformations act only on coordinates and partial derivatives.

	However, the approach that considers Lorentz transformations at the level of the gamma matrices becomes more appropriate in the generalized context of curved spacetimes. Indeed, the definition and properties of the adjoint spinor $\overline{\psi}=\psi^\dagger \gamma^{(0)}$ are crucial for maintaining consistency across curved geometries. The matrix $\gamma^{(0)}$ is fixed in the local tangent space as the canonical temporal gamma matrix in the Minkowski frame and does not transform as a vector, unlike the $\gamma^\mu$ matrices. The invariance of $\gamma^{(0)}$ is essential for constructing the adjoint spinor and preserving the indefinite inner product across all coordinate systems. More information can be found in \cite{saharian2020quantum}.

	\section{Morphism between pseudo-Riemannian and twisted spectral triples}
	\label{sec Presentation of the duality}
	In the previous sections, a link between twisted spectral triples and pseudo-Riemannian spectral triples has been revealed. This arises through the twist and the associated indefinite inner product, the role played by the corresponding unitaries in each context, and the real structure. This suggests the existence of a connection between these two generalized types of spectral triples. This section presents some important notions for what will follow and summarizes key results from \cite{nieuvsignchange} concerning this connection.

	\subsection{$\tw$-morphism of spectral triples}
	
	In this subsection, we present part of the results from \cite{nieuvsignchange}, specifically focusing on the definition of the $\tw$-morphism between the (even real) twisted spectral triple $(\calA, \calH, \Dir, J, \Gamma, \tw)$ and the pseudo-Riemannian spectral triples $(\calA, \calK, \KDir, J, \Gamma)$.

	Starting with the fundamental symmetry $\tw$, we construct the corresponding twist as $\rho(\, \cdot\,)=\tw (\, \cdot\,)\tw$. From the compatibility of the real structure with the twist and the result in equation \eqref{EqRelKGamma} we require \begin{align}
		J\tw=\epsilon \tw J,\qquad\qquad\qquad\qquad \tw \Gamma=\epsilon^\prime \Gamma\tw.
	\end{align}
	
	The $\tw$-morphism, denoted $\twphi$, is a bijective involutive map parametrized by $\tw$. It permits the transition from the twisted spectral triple $(\calA, \calH, \Dir, J, \Gamma, \tw)$ to a pseudo-Riemannian spectral triple $(\calA, \calK,\KDir, J, \Gamma)$ by constructing $\KDir$, the derivation and $\calK$ from $\Dir$ and $\calH$ through the following relations:
	\begin{align*}
		&\Dir\,\xrightarrow{\twphi}\, \KDir\defeq \tw \Dir,\\
		&[\Dir,  \, \cdot\,]_\rho\,\xrightarrow{\twphi}\, [\KDir, \, \cdot\, ],\\
		&\calH\defeq (E, \langle \, \cdot\, , \, \cdot\, \rangle)\, \xrightarrow{\twphi}\, \calK\defeq  (E, \inntw= \langle\, \cdot\, , \tw \cdot  \, \rangle),
	\end{align*}
	where $E$ is a vector space.
	
	The $\tw$-morphism does not act on $\calA$, $E$, $J$ and $\Gamma$ and preserves the evaluation of the Dirac operator by the inner product
	\begin{align*}
	Ev_{\,\psi}(\Dir, \inn)\defeq \langle \psi , \Dir\psi \rangle=\langle \psi , \KDir\psi \rangle_\tw=Ev_{\,\psi}(\KDir, \inntw)
	\end{align*}
for any $\psi\in E$.
	\begin{remark}
Following \cite{nieuvsignchange}, the definition of $\twphi$ can be extended to $\calA$ via the action of $\rho$ on $\calA$. It can also be extended to $\calH$ through the action of $\tw$ on $E$, thereby recovering the usual notions of $*$-algebra and vector-space morphisms. However, this does not affect the main results we will present. Therefore, we choose to define $\twphi$ without these additional actions, i.e., letting it act as the identity morphism on $\calA$ and $E$, in order to avoid the proliferation of unnecessary relations.
	\end{remark}
	As previously, the characteristic numbers of $(\calA, \calH, \Dir, J, \Gamma, \tw)$ are given by $\epsilon_0, \epsilon_1, \epsilon_2, \epsilon$, and $\epsilon^\prime$. Similarly, the characteristic numbers of $(\calA, \calK,\KDir, J, \Gamma)$ are given by $\epsilon_0^\tw, \epsilon_1^\tw, \epsilon_2^\tw,\epsilon$, and $\epsilon^\prime$, which lead to the following relations:
	$\epsilon_0=\epsilon_0^\tw$, $\epsilon_1^\tw=\epsilon\epsilon_1$, and $\epsilon_2^\tw=\epsilon_2$.

	By doing this it was checked in \cite{nieuvsignchange} that 
	\begin{itemize}
		\item The assertions $\Dir^\dagger =\Dir$ and $(\KDir)^\Tadj=\KDir$ implies each other.
		\item The derivations in each context are related by
		\begin{align}
			[\Dir, a]_\rho =\tw [\KDir, a]
		\end{align}
		for any $a\in\calA$. The $\tw$-morphism then preserve the boundedness of the derivation.
		\item The first order conditions in each contexts implies each other through the relation 
		\begin{equation}
			[[\Dir, a]_\rho, b^\circ]_{\rho^\circ}=\tw [[\KDir, a], b^\circ].
		\end{equation}  
 
		\item The fluctuations in each context are related by the relation $\Dir_{A_\rho}=\tw \KDir_{A^\tw}$ where
		\begin{align}
			\label{EqFluctKDir}
			&\KDir_{A^\tw}=\KDir+A^\tw+\epsilon_1^\tw J A^{\,\tw} J^{-1},\\
			\label{EqFluctTw}
			&\Dir_{A_\rho}=\Dir+A_\rho+ \epsilon_1  JA_\rho J^{-1},		
		\end{align}
		where $A^\tw$ is a connection of the form
		$A^\tw=\sum_ia_i[\KDir, b_i]$ and $A_\rho=\sum_i \rho(a_i)[\Dir, b_i]\rho$, a twisted connection such that $A^\tw=\tw A_\rho$.
		\item In particular, fluctuations that preserve the $\dagger$-selfadjointness of $\Dir$ and the $\Tadj$-selfadjointness of $\KDir$ are related by the $\tw$-unitary $U_\tw=u_\tw Ju_\tw J^{-1}$, with $u_\tw\in\calU_\tw(\calA)$
		\begin{align}
			\label{Eqfluct}
			\KDir_{A^\tw}\defeq	U_\tw\KDir U_\tw^\Tadj= \tw V_\tw \Dir  V_\tw^\dagger \defeq \tw \Dir_{A_\rho},
		\end{align}
		where $V_\tw=\rho(U_\tw)=\rho(u_\tw) J\rho(u _\tw) J^{-1}$ is also a $\tw$-unitary operator. The computational details can be found in \cite{nieuvsignchange}.
	\end{itemize}
	
	Note that all these results depend on the connection between the twist and the $\tw$-product, requiring the twist to be fundamental.
	
	The commutation relation between the grading and the Dirac operator cannot be the same for the two spectral triples. This is due to the presence of $\tw$, which connects $\Dir$ and $\KDir$. Specifically, defining $\epsilon_3^\tw=\pm 1$ such that $\KDir\Gamma=\epsilon_3^\tw\Gamma\KDir$, we have $\Dir\Gamma\defeq \epsilon_3\Gamma\Dir$ with $\epsilon_3=\epsilon^\prime\epsilon_3^\tw$. In the case where $\epsilon^\prime=-1$, we obtain a different relation between $\Gamma$ and the respective Dirac operators.
	
	If we require that $\epsilon_3^\tw=-1$, i.e., $\{\KDir, \Gamma\}=0$ (the usual case for spectral triples), then we find that $\Dir$ and $\Gamma$ satisfy a twisted version of the usual anti-commutator relation:
	\begin{equation}
		\{\Dir, \Gamma\}_\rho\defeq \Dir\Gamma+\rho(\Gamma)\Dir=\Dir\Gamma+\epsilon^\prime\Gamma\Dir=0.
	\end{equation} 
	In the manifold spectral triple's case, if $\epsilon_3=-1$, then we can work with the twist by grading on the algebra $\CM\otimes \mathbb{C}^2$, and the twist is trivial, i.e., $\rho(a)=a$ for any $a\in \CM\otimes \mathbb{C}^2$. However, if $\epsilon_3=1$, the boundedness condition on the twisted commutator will imply that we must work on the sub-algebra $\CM$, i.e., the sub-algebra corresponding to elements $a\in \CM\otimes \mathbb{C}^2$ on which the twist is trivial (see \cite{nieuvsignchange} for the proof). In this case, the $\tw$-morphism connects spectral triples for which the twist acts trivially on the algebra.
	
	We will focus on the case $\epsilon_3^\tw=-1$ (context 2 in \cite{nieuvsignchange}). This choice is motivated by the upcoming results and the fact that physics provides us with the Lorentzian $\KDir$ (a special case of $\KDir$ here), which anti-commutes with the grading. 
	
	As for the grading, the real structure $J=\tw\JH$ is preserved by the $\tw$-morphism. This is an important property of the $\tw$-morphism, as we prove that such a real structure is relevant both in the context of twisted spectral triples (see subsection \ref{SecTwistedRealStruct}) and in the one of pseudo-Riemannian spectral triples (see subsection \ref{SubsecCliff}).

	Now, let's see how this $\tw$-morphism is realized in the case of the spectral triple of even-dimensional manifolds.

	\subsection{Signature change for even-dimensional manifolds}
	\label{subsectionsignaturechange} 
	Let’s consider the $2m$-dimensional pseudo-Riemannian manifold $(\Man, \g)$ defined in subsection \ref{SubsecPseudoRiem}. The following subsection shows how the presented $\tw$-morphism acts in this context, changing the signature of the metric.
	
	Following definition \ref{DefReflection}, a reflection for $(\Man, \g)$ is an isometric automorphism $r$ on $\tm$ such that $r^2=\bbbone$.
	
	\begin{definition}[Spacelike reflection]
		Taking $(\Man, \g)$ to be a pseudo-Riemannian manifold, a spacelike reflection $r$ is a reflection for the metric $\g$ so that
		\begin{equation}
			\label{EqmetDef}
			\gr(\, \cdot\,,\, \cdot\, )\defeq \g(\, \cdot\,,\,r \, \cdot \, )
		\end{equation}
		is a positive definite metric on $\tm$.
	\end{definition}
	Taking the fundamental symmetry $\tw$ defined in \eqref{ExplicitK} with the associated twist $\rho$, we have that $\rho$ implements the spacelike reflection $r$ at the level of the Clifford representation
	\begin{equation}
		\label{Eqtwistparity}
		\rho(\cl(v))=\cl(rv)
	\end{equation}
	for any vector $v\in \tm$ (see \cite{nieuvsignchange}). The twist $\rho$ is then the parity operator as it reverses the spacelike coordinates, i.e., those corresponding to negative metric components. The spacelike property of $\rho$ is a direct consequence of the unitary requirement on the basis of the Clifford representation made in subsection \ref{SubsecCliff}.
	
	The corresponding spectral triple is given by $(\CM, \calK, \KDir, J, \Gamma)$\footnote{We choose to work on $\CM$, as the case $\CM\otimes \mathbb{C}^2$ does not add anything to the structures.} with $\calK=(L^2(\Man, \Sp), \inntw)$ and $\KDir, J$, and $\Gamma$ as defined in subsection \ref{SubsecPseudoRiem}, so that $\epsilon_3^\tw=-1$ i.e., 
	\begin{align}
		\{\KDir, \Gamma\}=0.
	\end{align} 
	
	This pseudo-Riemannian spectral triple will be taken as basis to construct a twisted spectral triple through the $\tw$-morphism. 
%	 This justifies the use of the space $L^2(\Man, \Sp)$ associated with the metric $\g$, with $L^2(\Man, \Sp)$ being preserved by the $\tw$-morphism. 
	
	Using equation \eqref{Eqtwistparity}, one can express relation \eqref{EqmetDef} at the Clifford algebra level
	\begin{equation}
		\label{Eqmetchg}
		2\gr(u, v)\mbb=2\g(u, rv)\mbb=\{\cl(u), \cl(rv)\}=\{\cl(u), \rho(\cl(v))\}. 
	\end{equation}
Noting that the action of $\twphi$ on $\KDir$ extend to the action
\begin{equation}
\cl(v)\xrightarrow{\twphi}\tilde{\cl}(v)\defeq \tw \cl(v)
\end{equation}
on the Clifford representation, this relation rephrase as
\begin{equation}
	2\gr(u, v)\mbb=\rho(\tilde{\cl}(u)\tilde{\cl}(v))+\tilde{\cl}(v)\tilde{\cl}(u). 
\end{equation}
We recover the two twisted Clifford relations presented in subsection \ref{SubsectwistCliffRel}.

	We can then define $\Dir=\tw\KDir$ within the corresponding twisted spectral triple $(\CM, \calH, \Dir, J, \Gamma, \tw)$, that satisfies $\epsilon_3=-\epsilon_3^\tw=1$, or equivalently
	\begin{equation}
		\{\Dir, \Gamma\}_{\rho}=0.
	\end{equation}  
	For standard (untwisted) spectral triples, the Riemannian metric $\gr$ is obtained in the distance formula from the norm evaluation of the derivation of elements in the algebra, providing the Lipschitz norm $\|grad\, a\|_\infty^2$ of $a\in\CM$ associated with $\gr$; see \cite{rivasseau2007non,connes1996gravity} for more information.
	
	Using this in the context of the twisted spectral triple $(\CM, \calH, \Dir, J, \Gamma, \tw)$ i.e., replacing the derivation with the twisted derivation in the distance formula, it was shown in \cite{nieuvsignchange} that $\|[\Dir, a]_\rho\|^2 =\|grad\, a\|_\infty^2$ implying that $(\CM, \calH, \Dir, J, \Gamma, \tw)$ can similarly be associated with a Riemannian metric $\gr$.

The $\tw$-morphism then acts on the metric as the transformation
	\begin{align}
		\label{EqMetChg}
		\g(\, \cdot\,,\, \cdot\, )\,\xrightarrow{\twphi}\, \gr(\, \cdot\,,\, \cdot\, )= \g(\, \cdot\,,\, r\,  \cdot \, )
	\end{align}
	thus changing its signature via the spacelike reflection $r$ associated with $\tw$.

Following \cite{nieuvsignchange}, we can also propose the metric formula
\begin{align}
	\label{EqMetTrace}
\g(u, v)=\frac{1}{2^m}\Tr(\cl(u)\cl(v))
\end{align}
on which the $\tw$-morphism acts at the level of the representations $\cl(v)\xrightarrow{\twphi}\tilde{\cl}(v)= \tw \cl(v)$ so that we similarly obtain the following rephrasing of \eqref{EqMetChg} 
\begin{align}
	\label{EqMetChangeLoc}
\g(u, v)=\frac{1}{2^m}\Tr(\cl(u)\cl(v))\,\xrightarrow{\twphi}\,\frac{1}{2^m}\Tr(\tcl(u)\tcl(v))=\frac{1}{2^m}\Tr(\cl(ru)\cl(v))=\gr(u, v).
\end{align}
The transformation \eqref{EqMetChg} therefore translates at the level of the Clifford representation 	
\begin{align}
	\g(u, v)\mbb=\frac{1}{2}\{\cl(u), \cl(v)\}\,\xrightarrow{\twphi}\,\frac{1}{2}(\rho(\tilde{\cl}(u)\tilde{\cl}(v))+\tilde{\cl}(v)\tilde{\cl}(u))=\gr(u, v)\mbb. 
\end{align}
%The generating vector space $V$ and the Clifford algebra are conserved by $\twphi$, but the relation to the metric is modified. It is then coherent with the fact that $\twphi$ preserve the operator $\Gamma$ and $J$ since they have been defined from the vector space $V$.
This leads to a generalization of the concept of Clifford relation presented in subsection \ref{SubsectwistCliffRel}. Two equivalent types of twisted Clifford relations can therefore be defined on the tangent space $T\Man$, starting from the representation $\cl$ or $\tilde{\cl}$
	\begin{align}
		\label{EqRepTw1}
		& 2\gr(u,v)\mbb=\cl(u)\rho(\cl(v))+\rho(\cl(v))\cl(u), \\
		\label{EqRepTw2}
		 &2\gr(u,v)\mbb=\tilde{\cl}(u)\tilde{\cl}(v)+\rho(\tilde{\cl}(v)\tilde{\cl}(u))
	\end{align} 
 where $u,v\in T_x\Man$ for a given $x\in\Man$. We will focus on the second one from now, as the representation $\tilde{\cl}$ will proves to be the relevant one for the Dirac operator $\Dir$.
 
	The relation to general coordinate $\mu$ is the same as in equation \eqref{EqRelVielb} since twisted Clifford representations are linear representations of the coordinates on which the vielbeins act. Note that the operations producing the signature change through the $\tw$-morphism are defined locally in equation \eqref{EqMetChangeLoc}, implying that the signature change is also performed locally, in the same manner as a Wick rotation in quantum field theory.
	
		Starting from $\KDir$ defined in \eqref{EqDefDirPseudo}, we define $\tgm^\mu=\tw \gamma^\mu$ so that $\Dir=-i \tgm^\mu\tilde\nabla_\mu^{{\scriptscriptstyle R, \Sp}}$ where
		\begin{equation}
			\label{EqParticGamma}
		\tilde\nabla_\mu^{{\scriptscriptstyle R, \Sp}}=\partial_\mu+\frac{1}{4}\tilde\Gamma^b_{\,\, \mu a}\tgm^a\tgm_b
		\end{equation}
	is the associated spin connection explicitly written in terms of the representation $\tilde{\cl}$, where the $\tilde\Gamma^b_{\,\, \mu a}$'s are defined from the $\Gamma^b_{\,\, \mu a}$'s in $\nabla_\mu^{{\scriptscriptstyle\Sp}}$ by the relation
	\begin{equation}
		\Gamma^b_{\,\, \mu a}=\g^a\g^b\tilde\Gamma^b_{\,\, \mu a}
	\end{equation}
using the fact that $\tgm^a\tgm_b=\tgm^a\tgm^b\g^{\scriptscriptstyle R}_b=\tgm^a\tgm^b=\tw\gamma^a\tw\gamma^b=\rho(\gamma^a)\gamma^b=\g^a\g^b\gamma^a\gamma_b$.

Note that $\tilde\nabla_\mu^{{\scriptscriptstyle R, \Sp}}$ is just a rewriting of $\nabla_\mu^{{\scriptscriptstyle\Sp}}$ in terms of the generator $\tilde{\cl}$. Doing so will become meaningful in what follows, when showing how the $\tilde\Gamma^b_{\,\, \mu a}$'s relate to the Christoffel symbols deduced from $\gr$.

The action of the covariant derivative for $\g$ on partial derivatives is given by $\nabla_\mu\partial_\nu=\Gamma^\lambda_{\,\, \mu \nu}\partial_\lambda$, where the $\Gamma^\lambda_{\,\, \mu \nu}$ are the (torsion-free) Christoffel symbols. The metric compatibility condition $\nabla_\nu\g_{\mu\kappa}=0$ implies that
\begin{equation}
	\label{EqMetFonda}
	\partial_\nu\g_{\mu\kappa}-\g_{\lambda\kappa}\Gamma^\lambda_{\,\, \mu \nu}-\g_{\mu\lambda}\Gamma^\lambda_{\,\, \kappa \nu}=0
\end{equation}
or equivalently
\begin{equation}
	\Gamma^\lambda_{\,\, \mu \nu}=\frac{1}{2}\g^{\lambda\kappa}(\partial_\mu\g_{\nu\kappa}+\partial_\nu\g_{\mu\kappa}-\partial_\kappa\g_{\mu\nu})
\end{equation}
For $\gr$ we have $\nabla_\mu^{{\scriptscriptstyle R}}\partial_\nu=\Gamma^\lambda_{{{\scriptscriptstyle R}}\, \mu \nu}\partial_\lambda$ together with the same connection between $\gr$ and $\Gamma^\lambda_{{{\scriptscriptstyle R}}\, \mu \nu}$ deduced from the metric compatibility condition $\nabla^{\scriptscriptstyle R}_\nu\g^{\scriptscriptstyle R}_{\mu\kappa}=0$.

 \begin{definition}[Reflected basis.]
	The "reflected" basis for $T^*\Man$ is defined by
	\begin{align}
		\label{EqReflect}
		dx^{r\mu}\defeq r\delta^\mu_\nu dx^\nu.
	\end{align}
\end{definition}
Since $r$ is an involution, its determinant satisfies $\det(r)=\pm 1$, implying that the orientation of the reflected basis is reversed when $\det(r)=-1$.

The fact that $\det(r^\mu_{\,\,\nu})\neq 0$ and that $r$ is both $\g$ and $\gr$ isometric ensure that the map $dx^{\mu}\to dx^{r\mu}$ is an isometric linear basis change, i.e. $\g^{r\mu r\nu}=\g(rdx^\mu, rdx^\nu)=\g^{\mu \nu}$ and $\gr^{r\mu r\nu}=\gr^{\mu \nu}$. Note that there exist coordinate functions $x^{r\mu}$ satisfying $dx^{r\mu}=d(x^{r\mu})$ only if the 1-form $dx^{r\mu}$ is integrable. In this case, the fact that $r^\mu_{\,\,\nu}$ is a fonction of $x\in\Man$ induces a nontrivial relation between $x^{r\mu}$ and $x^{\nu}$ because of the presence of the differential of $r$. If such a coordinate $x^{r\mu}$ exist, then equation \eqref{EqReflect} can be rewritten as a coordinate change 
\begin{align}
	dx^{r\mu}=\frac{\partial x^{r\mu}}{\partial x^\nu}dx^\nu=r\delta^\mu_\nu dx^\nu.
\end{align}
The dual basis is given by $\partial_{r\mu}=r\delta_\mu^\nu \partial_\nu$ as $\g(dx^{r\mu}, \partial_{r\nu})=\g(dx^{\mu}, \partial_{\nu})=\delta^\mu_\nu$ for all indices $\mu, \nu$. The corresponding Christoffel symbols are defined by
\begin{align}
	\nabla_\mu\partial_{r\nu}=\Gamma^{r\lambda}_{\,\, \mu r\nu}\partial_{r\lambda}.
\end{align}
The following proposition provides an important relation to connect the Christoffel symbols associated with the two metrics.
\begin{proposition}
	The Christoffel symbols $\Gamma^\lambda_{\,\, \mu \nu}$ and $\Gamma^\lambda_{{{\scriptscriptstyle R}}\, \mu \nu}$ for $\g$ and $\gr$ are related by
	\begin{align}
		\label{RelatChristos}
		\Gamma^{r\lambda}_{\,\, \mu r\nu}=\Gamma^\lambda_{{{\scriptscriptstyle R}}\, \mu \nu}+\frac{1}{2}\gr^{\lambda\kappa}(\partial_{r\nu}\g_{\mu\kappa}-\partial_{\nu}\g^{\scriptscriptstyle R}_{\mu\kappa}).
	\end{align}
\end{proposition}
\begin{proof}
	We have that 
	\begin{align*}
		\Gamma^{r\lambda}_{\,\, \mu r\nu}&=\frac{1}{2}\g^{{r\lambda}\kappa}(\partial_\mu\g_{r\nu\kappa}+\partial_{r\nu}\g_{\mu\kappa}-\partial_\kappa\g_{\mu r\nu})=\frac{1}{2}\gr^{\lambda\kappa}(\partial_\mu\g^{\scriptscriptstyle R}_{\nu\kappa}+\partial_{r\nu}\g_{\mu\kappa}-\partial_\kappa\g^{\scriptscriptstyle R}_{\mu \nu})\\
		&=\frac{1}{2}\gr^{\lambda\kappa}(\partial_\mu\g^{\scriptscriptstyle R}_{\nu\kappa}+\partial_{\nu}\g^{\scriptscriptstyle R}_{\mu\kappa}-\partial_\kappa\g^{\scriptscriptstyle R}_{\mu \nu})+\frac{1}{2}\gr^{\lambda\kappa}(\partial_{r\nu}\g_{\mu\kappa}-\partial_{\nu}\g^{\scriptscriptstyle R}_{\mu\kappa})\\
		&=\Gamma^\lambda_{{{\scriptscriptstyle R}}\, \mu \nu}+\frac{1}{2}\gr^{\lambda\kappa}(\partial_{r\nu}\g_{\mu\kappa}-\partial_{\nu}\g^{\scriptscriptstyle R}_{\mu\kappa})
	\end{align*}
	hence the result.
\end{proof}
This result shows how the $\tw$-morphism acts on the Christoffel symbols. Note that this permits to preserve the fact that each Christoffel symbol is a real number for any given triple of indices. This is not the case with Wick rotations where some Christoffel symbols become complex-valued.
\begin{remark}
	The relation \eqref{RelatChristos} permits to derive the metric condition for $\gr$ from the one of $\g$. Starting from equation \eqref{EqMetFonda} and turning the indices $\lambda$ and $\nu$ into $r\lambda$ and $r\nu$ (then reflecting the corresponding vectors), we obtain
	\begin{equation}
		\partial_{r\nu}\g_{\mu\kappa}-\g_{r\lambda\kappa}\Gamma^{r\lambda}_{\,\, \mu r\nu}-\g_{\mu r\lambda}\Gamma^{r\lambda}_{\,\, \kappa r\nu}=\partial_{r\nu}\g_{\mu\kappa}-\g^{\scriptscriptstyle R}_{\lambda\kappa}\Gamma^{r\lambda}_{\,\, \mu r\nu}-\g^{\scriptscriptstyle R}_{\mu \lambda}\Gamma^{r\lambda}_{\,\, \kappa r\nu}=0.
	\end{equation}
	Then using equation \eqref{RelatChristos}, this can be rewritten as
	\begin{equation*}
		\partial_{r\nu}\g_{\mu\kappa}-\g^{\scriptscriptstyle R}_{\lambda\kappa}\Gamma^\lambda_{{{\scriptscriptstyle R}}\, \mu \nu}-\g^{\scriptscriptstyle R}_{\mu \lambda}\Gamma^\lambda_{{{\scriptscriptstyle R}}\, \kappa \nu}-\frac{1}{2}\g^{\scriptscriptstyle R}_{\lambda\kappa}\gr^{\lambda\tau}(\partial_{r\nu}\g_{\mu\tau}-\partial_{\nu}\g^{\scriptscriptstyle R}_{\mu\tau})-\frac{1}{2}\g^{\scriptscriptstyle R}_{\mu\lambda}\gr^{\lambda\epsilon}(\partial_{r\nu}\g_{\kappa\epsilon}-\partial_{\nu}\g^{\scriptscriptstyle R}_{\kappa\epsilon})=0.
	\end{equation*}
	Using that $\g^{\scriptscriptstyle R}_{\mu\lambda}\gr^{\lambda\epsilon}=\delta^\epsilon_\mu$ this became
	\begin{equation*}
		\partial_{r\nu}\g_{\mu\kappa}-\g^{\scriptscriptstyle R}_{\lambda\kappa}\Gamma^\lambda_{{{\scriptscriptstyle R}}\, \mu \nu}-\g^{\scriptscriptstyle R}_{\mu \lambda}\Gamma^\lambda_{{{\scriptscriptstyle R}}\, \kappa \nu}-\frac{1}{2}(\partial_{r\nu}\g_{\mu\kappa}-\partial_{\nu}\g^{\scriptscriptstyle R}_{\mu\kappa}+\partial_{r\nu}\g_{\kappa\mu}-\partial_{\nu}\g^{\scriptscriptstyle R}_{\kappa\mu})=0
	\end{equation*}
	so that we retrieve the relation between $\gr$ and $\Gamma^\lambda_{{{\scriptscriptstyle R}}\, \mu \nu}$
	\begin{equation}
		\partial_{\nu}\g^{\scriptscriptstyle R}_{\mu\kappa}-\g^{\scriptscriptstyle R}_{\lambda\kappa}\Gamma^\lambda_{{{\scriptscriptstyle R}}\, \mu \nu}-\g^{\scriptscriptstyle R}_{\mu \lambda}\Gamma^\lambda_{{{\scriptscriptstyle R}}\, \kappa \nu}=0
	\end{equation}
	hence the metric compatibility condition $\nabla^{\scriptscriptstyle R}_\nu\g^{\scriptscriptstyle R}_{\mu\kappa}=0$.
\end{remark}
We are now interested in the expression of these relations in terms of the mixed indices basis $\mu, b, a$ appearing in the spin connection.

We can first note that $\g$-orthonormal bases are also $\gr$-orthonormal bases since $\g(dx^a, dx^b)=\g^{ab}=\delta^a_b\g^{ab}$ implies that
\begin{align*}
	\gr(dx^a, dx^b)=\g(rdx^a, dx^b)=\g(\delta^a_b\g^{ab}dx^a, dx^b)=\delta^a_b(\g^{ab})^2=\delta^a_b\gr^{ab}
\end{align*}
with $\gr^{aa}=(\g^{aa})^2\equiv (\g^{a})^2=1$ for any $a$. This implies that the vielbeins $e_a^\mu$ for $\g$ also correspond to vielbeins for $\gr$. The vielbeins $e_a^\mu$ can additionally be used to relate the reflected coordinate and orthonormal bases since
\begin{equation*}
dx^{ra}=re^{a}_{\mu}dx^{\mu}=e^{a}_{\mu}rdx^{\mu}=e^{a}_{\mu}dx^{r\mu}.
\end{equation*}

Note that from the fact that
\begin{equation}
	\rho(\tgm^a\tgm^b)=\rho(\tgm^a)\rho(\tgm^b)=\g^{a}\g^{b}\tgm^a \tgm^b,
\end{equation}
relation \eqref{EqRepTw2} can be rewritten as
\begin{equation}
	\label{EqCliffGeneralisé}
	2\gr^{ab}\mbb = \tgm^a\tgm^b+s_{ab}\tgm^b\tgm^a,
\end{equation}
where $s_{ab}=\g^{a}\g^{b}$ equals $1$ if $a=b$ and is $\pm 1$ if $a\neq b$. 

This can be seen as a direct generalization of the usual Clifford relations, where Clifford algebras correspond to the special case when $s_{ab}=1$ for all $a$ and $b$. 

In the $\g$-orthonormal basis, using that $dx^{ra}=rdx^a=\g^adx^a$ we get that
\begin{align}
	dx^{ra}=r\delta^a_b dx^b=\delta^a_b\g^b dx^b,
\end{align}
and consequently $\partial_{ra}=\delta^b_a\g_b\partial_b$.

As a consequence
\begin{align}
	\label{EqRewritTGamma}
\Gamma^{rb}_{\,\, \mu ra}=g^cg_d\delta^b_c\delta^d_a\Gamma^{c}_{\,\, \mu d}+g^c\delta^b_c\partial_\mu(g_c\delta^c_a)=g^bg_a\Gamma^{b}_{\,\, \mu a}+g^b\partial_\mu g_a=\tilde\Gamma^{b}_{\,\, \mu a}+g^b\partial_\mu g_a.
\end{align}
The relation between $\Gamma^b_{\,\, \mu a}$ and $\Gamma^\lambda_{\,\, \mu \nu}$ is given by $\Gamma^b_{\,\, \mu a}=e_a^\nu\Gamma^\lambda_{\,\, \mu \nu}e^b_\lambda-e_a^\nu\partial_\mu e^b_\nu$, this consequently extends to the mixed Christoffel symbols
\begin{align}
	\label{EqRelatGammVielb}
	\Gamma^{rb}_{\,\, \mu ra}=e_a^\nu\Gamma^{r\lambda}_{\,\, \mu r\nu}e^b_\lambda-e_a^\nu\partial_\mu e^b_\nu.
\end{align}
\begin{proposition}
The Dirac operator in the twisted spectral triple takes the form $\Dir=-i \tgm^\mu\tilde\nabla_\mu^{{\scriptscriptstyle R, \Sp}}$ where $\tilde\nabla_\mu^{{\scriptscriptstyle R, \Sp}}$ is the twisted version of the lift to the spinor bundle of the connection $\nabla_\mu^{{\scriptscriptstyle R}}$, which takes the form 
\begin{equation}
	\label{EqDefDir}
\tilde\nabla_\mu^{{\scriptscriptstyle R, \Sp}}=\partial_\mu+\frac{1}{4}\Gamma^b_{{{\scriptscriptstyle R}}\, \mu a}\tgm^a\tgm_b+\frac{1}{4}K^b_{\,\, \mu a}\tgm^a\tgm_b
\end{equation}
where $K^b_{\,\, \mu a}=g^b(\frac{1}{2}\g^{b\kappa}(\partial_{ra}\g_{\mu\kappa}-\partial_{a}\g^{\scriptscriptstyle R}_{\mu\kappa})-\partial_\mu g_a)$.
\end{proposition}
\begin{proof}
 Starting from the particular element $\tilde\Gamma^{b}_{\,\, \mu a}$ appearing in the expression of $\Dir$ in equation \eqref{EqParticGamma}, thanks to equations \eqref{EqRewritTGamma} and \eqref{EqRelatGammVielb} we have
\begin{align*}
	\tilde\Gamma^{b}_{\,\, \mu a}=\Gamma^{rb}_{\,\, \mu ra}-g^b\partial_\mu g_a=e_a^\nu\Gamma^{r\lambda}_{\,\, \mu r\nu}e^b_\lambda-e_a^\nu\partial_\mu e^b_\nu-g^b\partial_\mu g_a
\end{align*}
using equation \eqref{RelatChristos} and using the relation $\gr^{b\kappa}=\g^b\g^{b\kappa}$ at the last line, this rewrite
\begin{align*}
	\tilde\Gamma^{b}_{\,\, \mu a}&=e_a^\nu\Gamma^\lambda_{{{\scriptscriptstyle R}}\, \mu \nu}e^b_\lambda+e_a^\nu\frac{1}{2}\gr^{\lambda\kappa}(\partial_{r\nu}\g_{\mu\kappa}-\partial_{\nu}\g^{\scriptscriptstyle R}_{\mu\kappa})e^b_\lambda-e_a^\nu\partial_\mu e^b_\nu-g^b\partial_\mu g_a\\
	&=\Gamma^b_{{{\scriptscriptstyle R}}\, \mu a}+\frac{1}{2}\gr^{b\kappa}(\partial_{ra}\g_{\mu\kappa}-\partial_{a}\g^{\scriptscriptstyle R}_{\mu\kappa})-g^b\partial_\mu g_a\\
	&=\Gamma^b_{{{\scriptscriptstyle R}}\, \mu a}+g^b(\frac{1}{2}\g^{b\kappa}(\partial_{ra}\g_{\mu\kappa}-\partial_{a}\g^{\scriptscriptstyle R}_{\mu\kappa})-\partial_\mu g_a)
\end{align*}
hence the result.
\end{proof}
This provides an intrinsic definition of $\Dir$ in terms of the Christoffel symbols associated with $\gr$, together with the term $K^b_{\,\, \mu a}$ also function of $\gr$, since $\gr$ is deducible from $r$ and $\g$. The author did not find any interpretation or symmetry law for the $K^b_{\,\, \mu a}$ terms.

			\section{$\tw$-morphism from an almost-commutative twisted spectral triple}
			\label{sec Almost Commutative manifold}

			The previous section concluded with the construction of the twisted spectral triple $(\CM, \calH, \Dir, J, \Gamma, \tw)$ derived from the application of a $\tw$-morphism on a pseudo-Riemannian spectral triple. This twisted spectral triple will be taken as the starting point from now, so that $\Dir=-i \tgm^\mu\tilde\nabla_\mu^{{\scriptscriptstyle R, \Sp}}$ with $\tilde\nabla_\mu^{{\scriptscriptstyle R, \Sp}}$ defined in equation \eqref{EqDefDir}. 
			
			We will focus on the case corresponding to a pseudo-Riemannian spectral triple of KO-dimension 6. This choice is motivated by physical considerations, as it encompasses the 4-dimensional Lorentzian spectral triple case, which will be derived from the corresponding algebraic constraints in subsection \ref{Subsec4dimACmfld}. 
			
			The algebraic constraints of $(\CM, \calK, \KDir, J, \Gamma)$ are then given by
			\begin{equation}
				\label{EqParamPseudo}
				\epsilon_0^\tw=1,\qquad\qquad \epsilon_1^\tw=1,\qquad\qquad \epsilon_2^\tw=-1,\qquad\qquad \epsilon_3^\tw=-1.
			\end{equation}
			The numbers $\epsilon, \epsilon^\prime$ which depend on $n$ are left unfixed for the moment.
			
			The twisted spectral triple $(\CM, \calH, \Dir, J, \Gamma, \tw)$ obtained by the action of the $\tw$-morphism $\twphi$ is then specified by the four numbers
			\begin{equation}
				\label{DefEpsiTw}
				\epsilon_0=\epsilon_0^\tw,\qquad\qquad \epsilon_1=\epsilon\epsilon_1^\tw=\epsilon,\qquad\qquad \epsilon_2=\epsilon_2^\tw,\qquad\qquad \epsilon_3=\epsilon^\prime\epsilon_3^\tw=-\epsilon^\prime.
			\end{equation}  
			
			We now turn our attention to the almost-commutative manifold constructed on $\CM\otimes\AF$, where $\AF$ is a finite-dimensional algebra associated with the spectral triple $(\AF, \FDir, \HF, \JF, \GF)$, also of KO-dimension 6, whose manifold part is given by $(\CM, \calH, \Dir, J, \Gamma, \tw)$. The motivations to do this are the following:

			\begin{itemize}
				\item Explicit appearance of $\Dir$ in the Lorentzian Dirac Lagrangian in equation \eqref{EqLagDir} as
				\begin{equation}
					\overline{\psi} \KDir \psi=\psi^\dagger \Dir \psi.
				\end{equation} 
				
				\item The Noncommutative standard model is based on the product $\CM\otimes\AF$ with KO-dimension 6 for the finite part. A justification for this choice can be found in \cite{barrett2007lorentzian}, ensuring the elimination of the fermion doubling problem i.e., permitting us to obtain the physically good number of degrees of freedom for the fermionic fields. Masses are given by $\FDir$, corresponding to the second term of the Dirac Lagrangian. In the Lorentzian case, this mass term appears as the Lorentz scalar $m\psi^\dagger \gamma^{(0)} \psi$.

				\item In quantum field theory and noncommutative geometry, modular flow arises from the Tomita-Takesaki theory of von Neumann algebras. Given a state and an associated von Neumann algebra, the modular operator and the anti-linear modular conjugation define a time evolution on the algebra through the modular group, parametrized by a real parameter $t$. If the algebra is commutative, the flow is trivial. This provides an intrinsic interpretation of time, linking the flow of observables to the algebra's internal symmetries. 
				
				On another side, twisted spectral triples have been introduced in the context of von Neumann algebras where the twist itself is connected with the modular flow \cite{connes2006type}. In our approach, the twist is connected with the Krein structure of the dual pseudo-Riemannian spectral triple. The intuition is then that there may be a connection between the twist and a finite noncommutative part, hence the proposition to study the presented product structure with a noncommutative finite space.
			\end{itemize}
			The previous sections can be considered as a bottom-up approach, constructing a twisted spectral triple from a pseudo-Riemannian one. This section proposes a top-down reversed approach, where the starting point is now the presented twisted spectral triple within a product spectral triple. The aim will be to show how this very product structure harmonizes with the twisted spectral triple ones, to show how this implies the emergence of the dual pseudo-Riemannian spectral triple and to characterize this last one.

			\subsection{Product of twisted spectral triples and spectral triple}
			\label{SubsecProdAC}
			
			We start from the twisted spectral triple $(\CM, \calH, \Dir, J, \Gamma, \tw)$ with algebraic constraints given by $\epsilon_i=1$ for $i\in\{0, 1, 2, 3\}$ defined in \eqref{DefEpsiTw} and the numbers $\epsilon,\epsilon^\prime$. The finite spectral triple $(\AF, \FDir, \HF, \JF, \GF)$ of KO-dimension 6 is specified by the numbers $\epsilon_i^F=\epsilon_i^K$ for $i\in\{0, 1, 2, 3\}$ given in \eqref{EqParamPseudo}. The inner products on $\calH$ and $\HF$ are given by $\inn$ and $\inn_F$.

			The product spectral triple is denoted by $(\calA_p, \calH_p, \Dir_p, J_p, \Gamma_p, \tw_p)$ where the algebra is given by $\calA_p\defeq\CM\otimes \AF$ with Hilbert space $\calH_p\defeq\calH\otimes \HF$ and adjoint is $(\, \cdot\,)^{\adp} =(\, \cdot\,)^{\dagger}\otimes (\, \cdot\,)^{\dagger_F}$ where $\dagger_F$ is the adjoint associated with $\HF$. The full real structures, grading, and twist are given by 
			\begin{equation*}
				J_p\defeq J\otimes \JF\qquad\qquad \Gamma_p\defeq \Gamma\otimes  \GF \qquad\qquad \rho_{p}\defeq\rho\otimes \fbb
			\end{equation*}
			and the Dirac operator by 
			\begin{equation*}
				\Dir_p\defeq \Dir \otimes \fbb + \calO\otimes \FDir
			\end{equation*} 
			where $\calO$ is an operator that will be specified by the algebraic constraints connecting $J_p, \Dir_p$, and $\Gamma_p$ in the upcoming proposition. The operator's algebraic constraints are given by
			
			\begin{equation}
				\label{EqAlgConstr}
				J_p^2=\epsilon_0^p\quad\quad J_p\Dir_p =\epsilon_1^p\Dir_p J_p\quad\quad J_p\Gamma_p =\epsilon_2^p\Gamma_p J_p \quad\quad \Gamma_p \Dir_p = \epsilon_3^p\Dir_p \Gamma_p.
			\end{equation}

			The following proposition is an important result showing how the operator $\tw$ may appear naturally in $\Dir_p$.

			\begin{proposition}
				\label{propCstO}
				The spectral triple algebraic constraints (see \eqref{EqAlgConstr}) imply
				\begin{equation}
					\calO=\calO^\dagger\qquad  \qquad  J\calO=\epsilon\calO J \qquad  \qquad   \Gamma\calO=\epsilon^\prime\calO \Gamma.
				\end{equation} 
			\end{proposition}
			
			\begin{proof}
				The requirement $\Dir_p^\adp=\Dir_p$ implies directly $\calO=\calO^\dagger$.
				
				The "existence" of $\epsilon_1^p$ implies that
				\begin{equation*}
					J_p\Dir_p=J\Dir \otimes \JF+J\calO\otimes \JF\FDir =\epsilon_1\Dir  J\otimes \JF+\epsilon_1^FJ\calO\otimes \FDir  \JF
				\end{equation*}
				so that $\epsilon_1^p=\epsilon_1=\epsilon$ and then that $J\calO=\epsilon\epsilon_1^F\calO J=\epsilon\calO J$ since $\epsilon_1^F=1$.
				
				Then, the "existence" of $\epsilon_3^p$ implies that
				\begin{equation*}
					\Gamma_p\Dir_p= \Gamma\Dir \otimes \GF+ \Gamma\calO\otimes  \GF\FDir =\epsilon_3  \Gamma\Dir\otimes \GF+\epsilon_3^F \Gamma\calO\otimes \FDir  \GF
				\end{equation*}
				so that $\epsilon_3^p=\epsilon_3=-\epsilon^\prime$ and $ \Gamma\calO=-\epsilon^\prime\epsilon_3^F \calO \Gamma=\epsilon^\prime\calO \Gamma$ since $\epsilon_3^F=-1$.
			\end{proof}
			
			In addition, we have $\epsilon_0^p=\epsilon_0\epsilon_0^F=1$ and $\epsilon_2^p=\epsilon_2\epsilon_2^F=1$.
			
			In Proposition \ref{propCstO}, we see that $\calO$ satisfies all the properties required for the unitary operator $\tw$ (for compatibility with the real structure and the grading), while additionally imposing $\calO=\calO^\dagger$, i.e., exhibiting the property of a fundamental symmetry.
			
			\begin{proposition}
				\label{PropTwistInAC}
				The choice of operator $\calO=\tw^\prime$ such that $\rho(\, \cdot\,) =\tw^\prime(\, \cdot\,)\tw^\prime$ preserves the differential structure at the level of $(\AF, \FDir, \HF, \JF, \GF)$.
			\end{proposition}
			
			\begin{proof}
				Let us suppose that $\calO=\tw^\prime$ is a fundamental symmetry. Then, taking any $a=a_1\otimes a_2\in\AF$, the full differential is obtained by computing the composed twisted commutator
				\begin{align*}
					[\Dir_p, a]_{\rho_p}&=[\Dir , a_1]_{\rho}\otimes a_2+ \tw^\prime a_1\otimes \FDir  a_2 - \rho(a_1)\tw^\prime \otimes a_2\FDir
				\end{align*}
				which implies the requirement that $\tw^\prime a_1=\rho(a_1)\tw^\prime$ in order to obtain the term $[\FDir , a_2]$ in the full derivation
				\begin{equation}
					\label{EqFullDeriv}
					[\Dir_p, a]_{\rho_p}=[\Dir , a_1]_{\rho}\otimes a_2+ \tw^\prime a_1\otimes [\FDir , a_2].
				\end{equation}
				This requirement is equivalent to $\rho(a_1)=\tw^\prime a_1 \tw^\prime$.
			\end{proof}
			
			We therefore propose to replace $\calO$ by $\tw=\tw^\dagger$ so that we now consider the Dirac
			\begin{equation}
				\label{EqFullDir}
				\Dir_p=\Dir\otimes \fbb + \tw\otimes \FDir.
			\end{equation}
			
			This may appear surprising since it is now the algebraic structure of the almost-commutative manifold itself that invites us to introduce $\tw$, at the level of the finite Dirac operator. The conditions connecting $\tw$ to $J$ and $\Gamma$ through $\epsilon$ and $\epsilon^\prime$ are then induced by the spectral triple structure of the almost-commutative manifold. Conversely, starting from the unitary operator $\tw$ deduced in subsections \ref{subsecFromTwToIndef} and \ref{SubsecCliff}, the condition $\tw=\tw^\dagger$ ensuring the hermiticity of the product is now deduced from the requirement $\Dir_p^\adp=\Dir_p$. 
			
			This makes an interesting link between the product spectral triple structure and the twisted spectral triples, where the twist naturally emerges as a mechanism to decouple the differential structures. In addition, the term $\tw\otimes \FDir$ is exactly the expected term; having the form of a mass term.

			\begin{remark}
				Although the twist $\rho$ is trivial on $\CM$ (as any element in $\CM$ is proportional to identity), we keep the notation $\rho(a_1)$ to highlight the role played by the twist in full generality, i.e., for noncommutative generalizations of the twisted spectral triple $(\CM, \calH, \Dir, J, \Gamma, \tw)$ for which the twist is not necessarily trivial. In the case $\epsilon^\prime=-1$, as mentioned in Remark \ref{RQZZ}, the twist is the twist by the grading. Note that in this case, $\tw$-unitaries in $\CM$ are also unitary operators.
				
			\end{remark}
			
			We recover the twisted anticommutator relation for the product spectral triple
			\begin{equation}
				\{\Dir_p, \Gamma_p\}_{\rho_p}=0
			\end{equation}
			so that $\epsilon_3^p=-\epsilon^\prime$.
			
			\begin{remark}
				In the case $\epsilon=-1$, we have that $\epsilon_i$ corresponds to the inverse sign of the usual KO-dimension table in KO-dimension $4$, and that $\epsilon_i^p$ corresponds to the inverse sign of the KO-dimension $2$ for usual spectral triples. This may indicate that we can extend the KO-dimension table by permitting a more general relation between the grading and the Dirac with $\epsilon_3^p$, where the signs $\epsilon_i$ are reversed, as illustrated in the following table:
				
				\begin{table}[H]
					\centering
					\small
					\resizebox{0.5\textwidth}{!}{%
						\begin{tabular}{|c|*{4}{c|}|c|c|}
							\hline
							2m	& 0 & 2 & 4 & 6 & (2) & (4) \\
							\hline
							$\epsilon_0$ & 1 & -1 & -1 & 1 & 1 & 1 \\
							\hline
							$\epsilon_1$ & 1 & 1 & 1 & 1 & -1 & -1 \\
							\hline
							$\epsilon_2$ & 1 & -1 & 1 & -1 & 1 & -1 \\
							\hline
							$\epsilon_3$ & -1 & -1 & -1 & -1 & 1 & 1 \\
							\hline
						\end{tabular}
					}\normalsize
					\caption{\label{KoDimTable2} Parameters $\epsilon_0,\epsilon_1,\epsilon_2,\epsilon_3$ for KO-dimensions (mod 8).}
					\normalsize
				\end{table}
				
				The last two columns correspond to the two twisted spectral triples obtained in this section. The notation $(2)$ and $(4)$ is used to underline the fact that the KO-dimension interpretation is not clear in this twisted context. This point will be explored in future research. The apparent logic may then be the same as for the almost-commutative manifold of the noncommutative standard model of particle physics, i.e., we obtain a spectral triple of KO-dimension $(4+6)\bmod 8 =2$. 
			\end{remark}

			\subsection{$\tw$-morphism from the spectral triple product structure} 
			\label{SubsecEmergKisom}
			
			The following subsection aims at showing how the very presence of $\tw$ in equation \eqref{EqFullDir} induces the emergence of a Krein structure together with the pseudo-Riemannian spectral triple induced by the $\tw$-morphism.
			
			Thanks to this, by defining the $\tw$-selfadjoint Dirac operator $\Dir^\tw=\tw \Dir$, we can rewrite equation \eqref{EqFullDir} as
			\begin{equation}
				\label{EqformfactoDir}
				\Dir_p=\tw(\Dir^\tw\otimes \fbb + \mbb\otimes \FDir)
			\end{equation}
			and \eqref{EqFullDeriv} as
			\begin{align*}
				[\Dir_p, a]_{\rho_p}&=[\Dir , a_1]_{\rho}\otimes a_2+ \tw a_1\otimes [\FDir , a_2]\\
				&=\tw([\Dir^\tw, a_1]\otimes a_2+a_1\otimes [\FDir , a_2]).
			\end{align*}
			We recognize the first elements of transformation of a $\tw$-morphism for the Dirac and the corresponding derivation at the level of the first spectral triple.

			The adjoint $\Tadj$ is defined from $\dagger$ by the relation $(\, \cdot\,)^\Tadj=\rho(\, \cdot\,)^\dagger$. Let's now examine how $\tw$-unitaries $U_\tw\otimes \fbb$ arise as the natural operator to preserve the fundamental symmetry $\tw$ in $\Dir_p$, thus the term $\tw\otimes \FDir$ of the noncommutative part.

			General transformations that preserve the $\dagger$-selfadjoint property of $\Dir$ are given by the transformations $\Dir\,\to\, \tilde U\Dir \tilde U^\dagger$ where $\tilde U$ is not necessarily a unitary operator.

			Natural unitary transformations for the first twisted spectral triple must not affect the finite part term, namely $\tw$. The following proposition shows that they must therefore be implemented by $\tw$-unitary operators.

			\begin{proposition}
				The transformations of $\Dir_p$ by operators $\tilde{U}\otimes \fbb$ (at the level of the first spectral triple) given by
				\begin{equation}
					\Dir_p\,\to\,(\tilde{U}\otimes \fbb) \Dir_p (\tilde{U}^\dagger\otimes \fbb)\qquad\text{and}\qquad\Dir_p\,\to\,(\tilde{U}^\dagger\otimes\fbb) \Dir_p (\tilde{U}\otimes \fbb)
				\end{equation}
				that preserve $\tw\otimes \FDir$ (corresponding to the second spectral triple) are $\tw$-unitary operators.
			\end{proposition}
			
			\begin{proof}
				The transformation acts respectively on $\tw\otimes \FDir$ as
				\begin{equation}
					\tilde{U}\tw\tilde{U}^\dagger\otimes \FDir\qquad\text{and}\qquad\tilde{U}^\dagger\tw\tilde{U}\otimes \FDir.
				\end{equation}
				The term $\tw\otimes \FDir$ is then preserved if $\tilde{U}\tw\tilde{U}^\dagger=\tilde{U}^\dagger\tw\tilde{U}=\tw$ which implies $\tilde{U}\tilde{U}^\Tadj=\tilde{U}^\Tadj\tilde{U}=\mbb$, hence that $\tilde{U}$ is $\tw$-unitary.
			\end{proof}
			
			The $\tw$-unitary operators (for the first spectral triple) then arise naturally from the requirement that they must not affect the term $\tw\otimes \FDir$. This choice is reinforced by the following proposition, where this appears as a necessity to make the fluctuations of $\Dir$ and $\FDir$ independent.

			The following proposition shows how the fluctuations of the respective spectral triple are generated by generalized unitaries.
			\begin{proposition}
				\label{PropIndepFluctu}
				Taking the operators \( U_\tw\otimes U \) where \( U_\tw=u_\tw J u_\tw (J)^{-1} \) with $u_\tw\in \calU_\tw(\CM)$ and \( U=u \JF u (\JF)^{-1} \) with \( u \) a unitary operator in \( \AF \), the corresponding fluctuations of \( \Dir_p \) are given by
				\begin{equation}
					\label{EqFluctuDir}
					(U_\tw\otimes U) \Dir_p (U_\tw^\dagger \otimes U^{\dagger_F})= \tw(\Dir_{ A^\tw}^\tw \otimes \fbb + \mbb\otimes \Dir_{A_F})
				\end{equation}
				where 
				\begin{align*}
					&\Dir_{ A^\tw}^\tw\defeq V_\tw\Dir^\tw (V_\tw)^\Tadj=\KDir+A^\tw+\epsilon_1^\tw J A^{\tw} J^{-1}\\
					& \Dir_{A_F}\defeq U\FDir U^{\dagger_F}=\FDir+A_F+\epsilon_1^F \JF A_F \JF^{-1}
				\end{align*}
				with \( V_\tw=\rho(U_\tw) \) being also a \(\tw\)-unitary operator and the respective one forms
				\begin{align*}
					&A^{\tw}=v_\tw [\KDir, v_\tw^\Tadj]\\
					&A_F=u[\FDir, u^{\dagger_F}]
				\end{align*}
				with \( v_\tw=\rho(u_\tw) \).
			\end{proposition}
			\begin{proof}
				For the first term of equation \eqref{EqformfactoDir}, we have that 
				\begin{align*}
					(U_\tw\otimes U) \tw(\Dir^\tw\otimes \fbb) (U_\tw^\dagger \otimes U^{\dagger_F})&= \tw(\rho(U_\tw)\KDir U_\tw^\dagger \otimes \fbb )\\
					&=\tw(V_\tw \KDir V_\tw^\Tadj \otimes \fbb)
				\end{align*}

				For the second term of equation \eqref{EqformfactoDir}, we have
				\begin{align*}
					(U_\tw\otimes U)(\tw\otimes\FDir)(U_\tw^\dagger \otimes U^{\dagger_F})= \tw\otimes U\FDir U^{\dagger_F}
				\end{align*} 
				where the equality with the fluctuated Dirac operator \( \Dir_{A_F} \) is a well known result. 
			\end{proof} 
			
			Taking any $\psi=\psi_1\otimes \psi_2$ and $\psi^\prime=\psi^\prime_1\otimes \psi^\prime_2$ in $\calH_p$, the inner product $\langle \, \cdot\, , \, \cdot\, \rangle_p$ on $\calH_p$ is defined by the relation
			\begin{equation}
				\langle  \psi , \psi^\prime \rangle_p=\langle \psi_1 ,\psi^\prime_1 \rangle \langle \psi_2 ,\psi^\prime_2 \rangle_F.
			\end{equation}
			Using the $\tw$-product $\langle \, \cdot\, , \, \cdot\, \rangle_{\tw}= \langle \, \cdot\, , \tw \, \cdot\, \rangle$, the induced $\tw\otimes\fbb$-product for $\calH_p$ is then
			\begin{equation}
				\langle \, \cdot\, , \, \cdot\, \rangle_{\tw\otimes \fbb}\defeq \langle \, \cdot\, , \tw\otimes \fbb\, \cdot\, \rangle_p =\langle \, \cdot\, , \, \cdot\, \rangle_{\tw}\langle \, \cdot\, , \, \cdot\, \rangle_{F}.
			\end{equation}
			
			Then, defining $\KDir_p\defeq \tw\Dir_p=\Dir^\tw\otimes \fbb + \mbb\otimes \FDir$, the evaluation of $\Dir_p$ gives
			\begin{equation}
				\label{EqEvalDir}
				\langle  \psi , \Dir_p\psi^\prime \rangle_p=\langle  \psi , \KDir_p\psi^\prime \rangle_{\tw\otimes \fbb}=\langle \psi_1 ,  \Dir^\tw\psi^\prime_1 \rangle_{\tw} \langle \psi_2 ,\psi^\prime_2 \rangle_F+\langle \psi_1 , \psi^\prime_1 \rangle_{ \tw} \langle \psi_2 ,\FDir \psi^\prime_2 \rangle_F
			\end{equation}
			as $\tw$ is a fixed structure for the twisted spectral triple  $(\CM, \calH, \Dir, J, \Gamma, \tw)$. This relation also holds true for the fluctuated Dirac operators.   
			\begin{remark}
			Note that the fixed nature of $\tw$ is now induced by the fact that this operator is associated with the finite spectral triple in this top-down approach.
			\end{remark}
			We recognize the Krein structure at the level of the first spectral triple. By considering only the transformation at the level of the first spectral triple, this induces that only the general $\tw$-unitary operators $\hat{U}_\tw \otimes \fbb$ (not necessarily of the form $u_\tw J u_\tw (J)^{-1}\otimes \fbb$) acting on the structures of the first spectral triple as $\psi_1\,\to\, \hat{U}_\tw \psi_1$, $\psi^\prime_1\,\to\, \hat{U}_\tw \psi^\prime_1$ and $\Dir^\tw\,\to\, \hat{U}_\tw \Dir^\tw\hat{U}_\tw^\Tadj$ will preserve the evaluation \eqref{EqEvalDir}. This evaluation is more generally preserved by the $\tw\otimes\fbb$-unitary operators' actions at the level of the full spectral triple, with the generation of fluctuations in equation \eqref{EqFluctuDir} as a special case.

			In this way, the very structure of the product spectral triple singularizes a "sub-spectral triple" $(\CM, \calK,\KDir, J, \Gamma)$ where $\calK$ is obtained from $\calH$ by replacing its inner product with $\langle \, \cdot\, , \tw\, \cdot\, \rangle$. This pseudo-Riemannian spectral triple is nothing else that the $\tw$-morphism $\twphi$ applied on $(\CM, \Dir, \calH, J, \Gamma, \tw)$. This therefore induces that this product spectral triple structure contains and brings to the forefront a pseudo-Riemannian spectral triple. Note that these conclusions also hold if the twisted spectral triple on $\CM$ is replaced with a noncommutative generalization.

			\subsection{Structure of the induced pseudo-Riemannian spectral triple}
			\label{Subsec4dimACmfld}
			
			This subsection aims to present the method for determining the properties (particularly the signature) of the emerged pseudo-Riemannian spectral triple, with particular attention to the four-dimensional case.
			
			We start from the previous twisted spectral triple $(\CM, \Dir, \calH, J, \Gamma, \tw)$ for the $2m$-dimensional manifold $\Man$. The Dirac operator $\Dir$ takes the form $\Dir=-i \tgm^\mu\tilde\nabla_\mu^{{\scriptscriptstyle R, \Sp}}$ as defined in relation \eqref{EqDefDir}, where the $\tgm^\mu$ matrices are self-adjoint elements in $M_{2^m}(\mathbb{C})$ such that $\forall \mu$, $[\Gamma, \tgm^\mu]=0$. We additionally ask them to be unitaries in the $\gr$-orthonormal basis. From equation \eqref{EqMetTrace}, a general formula for the metric $\gr$ is given by the following:
			\begin{equation}
				\label{EqFty}
				\tr(\tgm^\mu\tgm^\nu)=2^m\gr^{\mu\nu}
			\end{equation}
			For the moment, we keep this general relation with the metric. A twisted Clifford-like relation will be provided at a later stage.
			
			We can then define $\kgm^\mu\defeq \tw \tgm^\mu$ such that we have $\KDir=\tw\Dir$. The subscript $\tw$ in $\kgm^\mu$ emphasizes the dependence of $\kgm^\mu$ on $\tgm^\mu$ through the choice of $\tw$ in this top-down approach. The purpose will then be to understand how the operator $\tw$ is constrained by the algebraic structure of the almost-commutative spectral triple and to characterize the resulting effects on the $\tgm^\mu$'s matrices and the associated metric.
			
			\begin{proposition}
				The $\kgm^a$'s are unitary operators, in general coordinate basis we have that
				\begin{equation}
					\label{Eqrelatgkmadj}
					(\kgm^\mu)^\Tadj=\kgm^\mu.
				\end{equation} 
			\end{proposition}
			\begin{proof}
				Since $\tgm^a$ is unitary we have $\tw\kgm^a (\kgm^a)^\dagger\tw=\mbb$ and then by multiplying by $\tw$ on the left and the right that $\kgm^a (\kgm^a)^\dagger=\mbb$. In the same way, $ (\kgm^a)^\dagger\tw\tw\kgm^a=\mbb$ implies directly $(\kgm^a)^\dagger\kgm^a=\mbb$ hence the result. Since $\tgm^\mu$ is self-adjoint, it follows that $\tw \kgm^\mu=(\kgm^\mu)^\dagger\tw$, leading to $\kgm^\mu=\tw(\kgm^\mu)^\dagger\tw\equiv(\kgm^\mu)^\Tadj$.
			\end{proof}
			
			In the same way, the operator $\KDir$ can be associated to a metric $\gk^{\mu\nu}$ defined by
			\begin{equation}
				\label{EqMetgk}
				\tr(\kgm^\mu\kgm^\nu)=2^m\gk^{\mu \nu}.
			\end{equation}
			Since $\tw^2=\mbb$, it follows that $\rho$ is an involution. Consequently, we obtain the splitting $M_{2^m}(\mathbb{C})=M_{2^m}(\mathbb{C})^+\oplus M_{2^m}(\mathbb{C})^-$, with $\rho(a^\pm)=\pm a^\pm$ for any $a^\pm\in M_{2^m}(\mathbb{C})^\pm$.

			The $\tgm^a$'s are related to the $\tgm^\mu$'s through the relation $\tgm^a=e_{\mu}^a\tgm^\mu$. Since the $\tgm^a$'s are mutually orthogonal, i.e.,$\tr(\tgm^a\tgm^b)=0$ for $a\neq b$, we choose the representation of the $\tgm^a$'s such that each $\tgm^a$ resides in either $M_{2^m}(\mathbb{C})^+$ or $M_{2^m}(\mathbb{C})^-$ 
			\begin{equation}
				\label{EqRHoGam}
				\rho(\tgm^a)\defeq s(a)\tgm^a\qquad\text{with}\qquad s(a)=\pm 1\quad \text{for}\quad \tgm^a\in M_{2^m}(\mathbb{C})^\pm.
			\end{equation}
			This will be a very convenient choice.

			The following proposition shows the connection of the signature of $\gr$ and $\gk$ in the local orthonormal basis.
			\begin{proposition}
				\label{PropChgMets}
				The metric $\gk^{ab}$ is $0$ if $a\neq b$ with signature $\gk^{a}=s(a)$ or equivalently
				\begin{equation}
					\label{EqRelMetri}
					\gk^{ab}=\gk^{a}\gr^{ab}.
				\end{equation}
			\end{proposition}
			\begin{proof}
				equation \eqref{EqRHoGam} is equivalent to
				\begin{equation}
					\label{EqRLM}
					\rho(\kgm^a)=s(a)\kgm^a.
				\end{equation}
				Then we have 
				\[
				2^m\gr^{ab}=\tr(\tgm^a \tgm^b )=\tr(\rho(\kgm^a)\kgm^b )=s(a) \tr(\kgm^a \kgm^b )=s(a)2^m\gk^{ab}
				\]
				and then $\gr^{ab}=s(a)\gk^{ab}$ which implies that $\gk^{ab}=0$ for $a\neq b$ and $\gr^{a}=s(a)\gk^{a}$ so that $s(a)=\gk^{a}$ and $\gr^{ab}=\gk^{a}\gk^{ab}$.
			\end{proof}
			We recover the result from relation \eqref{EqRelMetRho}, i.e., $\rho(\kgm^a)=\gk^{a}\kgm^a$. The number of "temporal" directions $n$ is then defined as $\card\{a\, \mid\, \gk^{a}=1\}$, and the grading $\Gamma$ is obtained from the $\kgm^a$'s through equation \eqref{DefOpGrad}.
			
%			The relation $	\tr(\tgm^\mu\tgm^\nu)=\tr(\rho(\kgm^\mu)\kgm^\nu)$ does not permit in itself to obtain a concrete link between $\gr$ and $\gk$ in general coordinates, as $\tgm^\mu$ and $\kgm^\mu$ are not necessarily unitaries. This is obtained by passing through the local orthonormal basis, see the following proposition.
%			
%			
%			\begin{proposition}
%				$\gk^{\mu \nu}$ and $\gr^{\mu \nu}$ are related by 
%				\begin{equation*}
%					\gr^{\mu\nu}=\gk^{a}  e_{R, a}^\mu e_{R, b}^\nu e_{\tw, \lambda}^a  e_{\tw, \kappa}^b \gk^{\lambda \kappa}
%				\end{equation*}
%				where $e_{\tw, \lambda}^a$'s are the vielbeins for the metric $\gk$.\GN{changer les vielbeins}
%			\end{proposition}
%			\begin{proof}
%				We have that
%				\begin{align*}
%					2^m\gr^{\mu\nu}&=e_{R, a}^\mu e_{R, b}^\nu\tr(\tgm^a\tgm^b)=\gk^{a} e_{R, a}^\mu e_{R, b}^\nu \tr(\kgm^a \kgm^b )=2^m\gk^{a}  e_{R, a}^\mu e_{R, b}^\nu e_{\tw, \lambda}^a  e_{\tw, \kappa}^b \gk^{\lambda \kappa}
%				\end{align*}
%			\end{proof}
		 
			We now require the $\tgm^a$'s to satisfy the twisted Clifford relation of equation \eqref{EqRepTw2} 
			\begin{equation}
				\label{EqMetgtw}
				2\gr^{ab}\mbb = \tgm^a\tgm^b+\rho(\tgm^b\tgm^a).
			\end{equation} 
			This can be interpreted as a compatibility condition between $\tw$ and the $\tgm^a$ matrices, which becomes necessary to obtain the following Clifford algebraic relation
			\begin{equation}
				\label{EqMetgkcli}
				2\gk^{ab}\mbb = \kgm^a\kgm^b+\kgm^b\kgm^a,
			\end{equation}
			to connect the dual $\kgm^a$'s with the metric $\gk$.
			
			The relation between $\gr$ and $\gk$ at the level of the Clifford algebraic structures in equations \eqref{EqMetgtw} and \eqref{EqMetgkcli} is given by observing that equation \eqref{EqMetgtw} is equivalent to $\rho(\kgm^a)\kgm^b+\kgm^b\rho(\kgm^a)$ so that 
			\begin{align*}
				2\gr^{ab}\mbb &= \rho(\kgm^a)\kgm^b+\kgm^b\rho(\kgm^a)=\gk^{a} (\kgm^a\kgm^b+\kgm^b\kgm^a)=2\gk^{a}	\gk^{ab}\mbb,
			\end{align*}
			using equation \eqref{EqRelMetri}.

			We have in KO-dimension 6 that $2(n-m)=8N+6$ where $N\in\mathbb{Z}$, implying that $n=4N+m+3$. This shows that $n$ is odd if $m$ is even, which further implies that $\epsilon^\prime=-1$ (from proposition \ref{PropRelKJG}), and that $n$ is even if $m$ is odd, leading to $\epsilon^\prime=1$. 
			
			We now focus on the four-dimensional case. Since $m$ is even, it follows that $\epsilon^\prime=-1$ and $n$ is odd. The following proposition presents a significant result of this article, demonstrating how the algebraic constraints of the spectral triple structure of the almost-commutative Manifold are sufficiently strong to explicitly determine $\tw$ and the signature of the induced metric. 
			
			\begin{proposition}
				In dimension 4, the fundamental symmetry $\tw$ induces the transformation from the metric $\gr$ with signature $(+, +, +, +)$ to the metric $\gk$ with signature $(+, -, -, -)$ when $\epsilon=-1$, and to $(+, +, +, -)$ when $\epsilon=1$.
			\end{proposition}

			\begin{proof}
				We aim to demonstrate that the conditions
				\begin{equation}
					\tw^\dagger=\tw,\qquad \qquad \tw \Gamma=-\Gamma\tw, \qquad \qquad \tw J=\epsilon J\tw 
				\end{equation}
				which arise from the structure of the spectral triple of the almost-commutative manifold, together with the relation $\rho(\kgm^a)=\gk^{a}\kgm^a$, implies that only one index $a$ for which $\gk^{a}=1$ when $\epsilon=-1$, and exactly three such indices when $\epsilon=1$.

				We have that the matrices $\Gamma^a\defeq\alpha_a\kgm^{(a)}$ where $\alpha_a$ is a complex number such that $(\Gamma^a)^\dagger=\Gamma^a$ and $(\Gamma^a)^2=\bbbone_4$ constitute a free family due to the metric relation, see equation \eqref{EqMetgkcli}. This free family can be extended by incorporating the elements $\bbbone_4$ and the grading operator $\Gamma$ (see equation \eqref{DefOpGrad}), together with the elements 
				\begin{align*}
					&\Gamma^{ab}\defeq \alpha_{ab}\kgm^{(a)}\kgm^{(b)}\qquad\qquad\,\,\,\,\text{with}\qquad a<b\\
					&\Gamma^{abc}\defeq \alpha_{abc}\kgm^{(a)}\kgm^{(b)}\kgm^{(c)}\qquad\quad\text{with}\qquad a<b<c
				\end{align*}
				
				with $\alpha_{ab}$ and $\alpha_{abc}$ being complex numbers such that $(\Gamma^{ab})^\dagger=\Gamma^{ab}$ and $(\Gamma^{ab})^2=\bbbone_4$, and similarly $(\Gamma^{abc})^\dagger=\Gamma^{abc}$ and $(\Gamma^{abc})^2=\bbbone_4$.  
				
				There are $6$ elements $\Gamma^{ab}$ and $4$ elements $\Gamma^{abc}$, such that the set $\calI\defeq\{\bbbone_4, \{\Gamma^a\}_a, \{\Gamma^{ab}\}_{a<b}, \{\Gamma^{abc}\}_{a<b<c}, \Gamma\}$ forms a free family of $16$ orthogonal elements\footnote{The metric properties ensure the freeness of the family, implying that the trace of any product of two distinct elements in $\calI$ vanishes, thus guaranteeing orthogonality.}. 
				
				Since $\calI$ is a free family of 16 elements, it is therefore also a basis of $ M_4(\mathbb{C})$ so that we can write $\tw\in  M_4(\mathbb{C})$ according to this basis
				\begin{equation}
					\tw=\lambda\bbbone_4 + e_a\Gamma^a+ e_{bc}\Gamma^{bc}+ e_{def}\Gamma^{def}+ \epsilon \Gamma
				\end{equation}
				where $\lambda, e_a, e_{bc}, e_{def}, \epsilon$ are complex numbers.
				
				The condition $\tw \Gamma=-\Gamma\tw$ then implies that $\lambda, e_{ab}$ and $\epsilon$ must be zero because $\Gamma\kgm^{(a)}=-\kgm^{(a)}\Gamma$ for any $a$. Consequently, $\tw$ can be expressed as follows
				\begin{equation}
					\tw=e_a\Gamma^a+e_{bcd}\Gamma^{bcd}.
				\end{equation}

				Now using the fact that $\tw\kgm^e=\gk^{e}\kgm^e\tw$, we have that
				\begin{equation*}
					(e_a\Gamma^a+e_{bcd}\Gamma^{bcd})\kgm^e = \gk^{e}\kgm^e (e_a\Gamma^a+e_{bcd}\Gamma^{bcd}).
				\end{equation*}
				We are left with two possibilities
				\begin{itemize}
					\item If $\gk^{e}=1$ for a given (fixed) index $e$, we then have that $e_a=0$ if $a\neq e$ and $e_{bcd}=0$ if $e\notin \{b, c, d\}$. As $e$ is a fixed index, there remain three distinct fixed indices $a_1,a_2,a_3 $ so that $a_i\neq e$ for $i=1, 2, 3$. In this way, $\tw$ takes the form
					\begin{equation}
						\tw=e_e\Gamma^e+e_{(ea_1a_2)}\Gamma^{(ea_1a_2)}+e_{(ea_1a_3)}\Gamma^{(ea_1a_3)}+e_{(ea_2a_3)}\Gamma^{(ea_2a_3)}
					\end{equation}
					where the notation $(abc)$ in $e_{(abc)}\Gamma^{(abc)}$ indicates that the indices $a, b, c$ are ordered in increasing magnitude, e.g., if $b<c<a$ then $e_{(abc)}\Gamma^{(abc)}=e_{bca}\Gamma^{bca}$.
					
					Now if we consider the relation $\tw\kgm^{a_1} =\gk^{a_1} \kgm^{a_1} \tw$, this very relation implies that either $e_{(ea_2a_3)}=0$ or $e_{(ea_1a_2)}=e_{(ea_1a_3)}=0$ according to the sign of $\gk^{a_1}$. In the case $e_{(ea_1a_2)}\neq 0$ and $e_{(ea_1a_3)}\neq 0$ we can now consider the relation $\tw\kgm^{a_2} =\gk^{a_2} \kgm^{a_2} \tw$, this implies that either $e_{(ea_1a_2)}=0$ or $e_{(ea_1a_3)}= 0$ so that in any case, only one of the $3$ coefficients $e_{(ea_1a_2)}, e_{(ea_1a_3)}, e_{(ea_2a_3)}$ will be non-zero. Thus, $\tw$ ultimately takes the form 
					\begin{align}
						\label{EqFormKa}
						\tw=e_e\Gamma^e+e_{(ea_1a_2)}\Gamma^{(ea_1a_2)}
					\end{align}
					
					\item If $\gk^{e}=-1$ for a given (fixed) $e$, we then have that $e_a=0$ if $a= e$ and $e_{bcd}=0$ if $e\in \{b, c, d\}$. We then have 
					\begin{equation}
						\tw=e_{a_1}\Gamma^{a_1}+e_{a_2}\Gamma^{a_2}+e_{a_3}\Gamma^{a_3}+e_{(a_1a_2a_3)}\Gamma^{(a_1a_2a_3)}
					\end{equation}
					where $a_i\neq e$.

					Now if we consider the relation $\tw\kgm^{a_1} =\gk^{a_1}\kgm^{a_1} \tw$, then either $e_{a_1}=0$ or $e_{a_2}=e_{a_3}=0$. In the case $e_{a_2}\neq 0$ and $e_{a_3}\neq 0$, considering the relation $\tw\kgm^{a_2} =\gk^{a_2} \kgm^{a_2} \tw$ leads to the fact that either $e_{a_2}=0$ or $e_{a_3}=0$ so that in any case, only one of the indices $e_{a_i}$ is non-zero. Then $\tw$ takes the form
					\begin{equation}
						\label{EqFormKb}
						\tw=e_{a_1}\Gamma^{a_1}+e_{(a_1a_2a_3)}\Gamma^{(a_1a_2a_3)}
					\end{equation}
					which is equivalent to \eqref{EqFormKa}.
				\end{itemize}

				Then, in any case, $\tw$ is now given by \eqref{EqFormKb}. Now let's consider again $\tw\kgm^{a_2} =\gk^{a_2} \kgm^{a_2} \tw$ on this last expression, this clearly implies that either $e_{a_1}=0$ or $e_{(a_1a_2a_3)}=0$. In all cases, the condition $\tw^\dagger=\tw$ leads to either $e_{a_1}=1$ or $e_{(a_1a_2a_3)}=1$.

				\begin{itemize}
					\item In the case $\tw=\Gamma^{a_1}$, we have that $\rho(\kgm^{a_1})=\kgm^{a_1}$. Then knowing that $\rho(\kgm^a)=(\kgm^a)^\dagger$, we get that $(\kgm^{a_1})^\dagger=\kgm^{a_1}$ inducing that $\tw=\Gamma^{a_1}=\kgm^{(a_1)}$. This is only consistent with the case $\epsilon=-1$ because $J\kgm^{(a)}=-\epsilon_1^\tw\kgm^{(a)}J=-\kgm^{(a)}J$. Consequently, $\rho(\kgm^a)=\gk^{a}\kgm^a=\kgm^a$ holds exclusively for $a=a_1$, leading to the signature $(+, -, -, -)$ for $\gk$.

					\item Now if $\tw=\Gamma^{a_1a_2a_3}$ with $a_1<a_2<a_3$, for the same reason we have that $\kgm^{a_1}, \kgm^{a_2}, \kgm^{a_3}$ are selfadjoint. Then the requirement $\tw=\tw^\dagger$ becomes
					\begin{equation*}
						\alpha_{a_1 a_2 a_3}\kgm^{(a_1)}\kgm^{(a_2)}\kgm^{(a_3)}=\alpha^*_{a_1 a_2 a_3}\kgm^{(a_3)}\kgm^{(a_2)}\kgm^{(a_1)}=-\alpha^*_{a_1 a_2 a_3}\kgm^{(a_1)}\kgm^{(a_2)}\kgm^{(a_3)}
					\end{equation*}
					inducing that $\alpha_{a_1 a_2 a_3}=i$. Then the condition $J\kgm^{(a)}=-\kgm^{(a)}J$ implies that 
					\begin{equation*}
						J \tw=J (i\kgm^{(a_1)}\kgm^{(a_2)}\kgm^{(a_3)})=-i(-1)^3 \kgm^{(a_1)}\kgm^{(a_2)}\kgm^{(a_3)} J =\tw J.
					\end{equation*}
					This is consistent only with the case $\epsilon=1$. Consequently, $\rho(\kgm^a)=\gk^{a}\kgm^a=\kgm^a$ holds exclusively for $a\in\{a_1,a_2,a_3\}$, leading to the signature $(+, +, +, -)$ for $\gk$.
				\end{itemize}
			\end{proof}
			
			This structure consequently leads to the uniqueness of the time dimension in four dimensions, as the two cases $\epsilon=\pm 1$ are merely different conventions when time is associated with the negative metric component in the case $\epsilon=1$. This is a very specific and welcoming result, revealing the connection between the signature and the algebraic constraints of the almost-commutative spectral triple.

			In the case $\epsilon=-1$, which corresponds to the standard Lorentzian case described at the end of subsection \ref{SubsecPseudoRiem}, renaming $a_1$ as $0$ leads to $\tgm^{(0)}=\bbbone_4$ and $\tw=\kgm^{(0)}$, as we obtained in subsection \ref{SubsecPseudoRiem}. In this case, the evaluation of $\Dir_p$ in equation \eqref{EqEvalDir} is nothing else than the Dirac Lagrangian. This is a scalar (invariant) under Lorentz transformations, being $\tw$-unitary operators. We retrieve all the necessary components for a Lorentzian version of the Standard Model (apart from the spectral action), notably the correct form for the mass term given by $\tw\otimes \FDir$. Note that the twisted spectral triple used to deduce the Lorentzian one has also been used in the context of the noncommutative standard model of particle physics in \cite{bochniak2020spectral}, recalling that the twist acts trivially on $\CM$.

			\newpage
			\section{Conclusion}
			The role played by $\tw$ is fundamental in most of the presented structures. Twisted spectral triples and pseudo-Riemannian spectral triples appear as generalizations of the usual spectral triple, recovered in the special case $\tw=\bbbone$. This is precisely at this level that the $\tw$-morphism allows the desired signature change. The remaining question was why this may occur. This very question motivates the study of the so-called top-down approach presented in section \ref{sec Almost Commutative manifold}, which places the twisted spectral triple within a product spectral triple, following the spirit of the noncommutative standard model of particle physics. 
			
			The main assertion of this article has been the proposition that pseudo-Riemannian spectral triples may be seen as emergences from twisted spectral triples within an almost-commutative spectral triple structure. The KO-dimension of the emerged pseudo-Riemannian spectral triple and that of the finite noncommutative one are the same. In the end, the Almost-commutative spectral triple can be seen as a first given object, where the twisted structure arises as a consequence of the "entanglement" between the two spectral triples of the product through the operator $\tw$. The special twisted Clifford algebraic structure then appears as a modified way to algebraically represent the tangent space $\tm$ and retrieve the metric within the algebraic relations.

			In the four-dimensional case, we recover $\twx=\gamma^{(0)}(x)\defeq \gamma^a|_{a=0}(x)$, the fundamental symmetry which appears in Lorentzian quantum field theory, specific to the signature $(+, -, -, -)$. This very operator is linked to the finite noncommutative geometry of the Almost-commutative spectral triple and acts as a local fixed point in the theory, providing the mass term and enabling the definition of the desired coordinate-independent Krein structure, unsensitive to Lorentz transformations.
			
			This article focuses on special cases of physical interest. We leave the generalization of these results for future study. This will, in particular, concern the study of what happens in the odd-dimensional manifold case, the exploration of more general KO-dimensions (not only the KO-dimension 6), and the characterization of the emerged signature in these more general contexts. An exploration of the potential connection with Tomita-Takesaki theory is also envisioned.

			We hope the presented structure will offer concrete avenues for the study of General Relativity within the framework of spectral triples. 	The specific algebraic constraints arising from the noncommutative geometrical formulation of Yang-Mills-Higgs theory are algebraically linked to the causal structure of spacetime through the fundamental symmetry operator $\tw$. This conceptual connection may have significant physical implications that I intend to investigate in future works.

			\section*{Acknowledgements}
			
			I would like to express my sincere gratitude to T. Masson, S. Lazzarini, and A. Sitarz for the insightful discussions and valuable exchanges I had with them throughout the development of this work.

			%
			%\section*{Conflict of Interest Statement}
			%
			%The authors declare that there are no conflicts of interest regarding the publication of this paper.
			%
			%\section*{Data Availability Statement}
			%The data that support the findings of this study are available from the corresponding author upon reasonable request. 

			\newpage
			\bibliography{bibliography}

		\end{document}